\documentclass[journal,twocolumn]{IEEEtran}
\usepackage{amsmath,graphicx,amsfonts,amsthm}
\usepackage{epstopdf}
\usepackage{xcolor,colortbl}
\usepackage{verbatim}
\usepackage{amsmath,bm}
\usepackage{lipsum}
\usepackage{color}
\usepackage{cite}
\usepackage{enumitem}
\usepackage{subcaption}

\def\bfh{{\mathbf h}}
\def\bfn{{\mathbf n}}

\def\bfg{{\mathbf g}}

\def\bfy{{\mathbf y}}

\def\cR{{\mathbf R}}
\def\SINR{\mathrm{SINR}}

\def\cC{{\cal C}}

\def\cR{{\cal R}}
\def\E{\mathbb{E}}

\newtheorem{theorem}{Theorem}

\usepackage{amsthm}
\newtheorem*{theorem*}{Theorem}
\newtheoremstyle{named}{}{}{\itshape}{}{\bfseries}{.}{.5em}{\thmnote{#3}}
\theoremstyle{named}
\newtheorem*{namedtheorem}{Theorem}

\usepackage{algpseudocode, algorithm}

\begin{document}

\title{Random Pilot and Data Access in Massive MIMO for Machine-type Communications}

\IEEEoverridecommandlockouts

\author{Elisabeth de Carvalho,~\IEEEmembership{Senior Member,~IEEE,}  
Emil Bj{\"o}rnson,~\IEEEmembership{Member,~IEEE,}
Jesper H. S\o{}rensen,~\IEEEmembership{Member,~IEEE,}\\
        Erik G. Larsson,~\IEEEmembership{Fellow,~IEEE,}
        Petar Popovski,~\IEEEmembership{Fellow,~IEEE}
\thanks{E.~de Carvalho, J.~H.~S\o{}rensen, and P.~Popovski are with the Department of Electronic Systems, Aalborg University, Aalborg, Denmark (email: edc@es.aau.dk, jhs@es.aau.dk, petarp@es.aau.dk).
E.~Bj\"ornson and E.~G.~Larsson are with the Department of Electrical Engineering (ISY), Link\"{o}ping University, Link\"{o}ping, Sweden (email: emil.bjornson@liu.se, erik.g.larsson@liu.se).
}%
\thanks{This work was performed partly in the framework of the Danish Council for Independent Research (DFF133500273),  the Danish National Advanced Technology Foundation via the VIRTUOSO project,
the Horizon 2020 project FANTASTIC-5G (ICT-671660), the EU FP7 project MAMMOET (ICT-619086), ELLIIT, and CENIIT. The authors would like to acknowledge the contributions of the colleagues in FANTASTIC-5G and MAMMOET.
}%
\thanks{This work was presented in part at the IEEE International Conference on Acoustics, Speech and Signal Processing (ICASSP), Shanghai, China,  2016.}%
}

\maketitle

\begin{abstract}
A massive MIMO system, represented by a base station with hundreds of antennas, is capable of spatially multiplexing many devices and thus naturally suited to serve dense crowds of wireless devices in emerging applications, such as machine-type communications. Crowd scenarios pose new challenges in the pilot-based acquisition of channel state information and call for pilot access protocols that match the intermittent pattern of device activity. A joint pilot assignment and data transmission protocol based on random access is proposed in this paper for the uplink of a massive MIMO system. The protocol relies on the averaging across multiple transmission slots of the pilot collision events that result from the random access process. 
We derive new uplink sum rate expressions that take pilot collisions, intermittent device activity, and interference into account. Simplified bounds are obtained and used to optimize the device activation probability and pilot length. A performance analysis indicates how performance scales as a function of the number of antennas and the transmission slot duration. 
\end{abstract}
%

\section{Introduction}

In a massive multiple-input multiple-output (MIMO) system~\cite{Marzetta2010a}, 
the base station (BS) is equipped with a very large number of antennas that create a very large number of spatial degrees of freedom (DoFs) under the asymptotic favorable propagation conditions that most channels seem to offer~\cite{Bjornson2016b,MDN2014,Ngo2014a}. 
Those spatial DoFs can be used in different ways. 
If the number of antennas at the BS is much larger than the number of wireless devices~\cite{Marzetta2010a},
then the extra DoFs are exploited to generate strong spatial beams to the devices, which are in effect hardened communication channels with negligible small-scale fading and a minor multi-user interference.
When the number of devices is large, the DoFs can be used for aggressive spatial multiplexing of devices, which does bring the largest cell spectral efficiency~\cite{Bjornson2016a}. However, the devices' channels are not asymptotically decorrelated in this operating regime and the spectral efficiency per device is thus lowered. 
The enhanced spatial resolution of massive MIMO is essential in \emph{crowd} scenarios, such as shopping mall or a stadium \cite{5GBookMetis}, serving high density of devices that are closely spaced. The scenario that sets the motivation for this article is the one of massive Machine-Type Communication (mMTC), where around $100\,000$ devices could be connected to a single BS. 

Acquisition of channel state information (CSI) in massive MIMO is critical, as a huge number of channel 
coefficients have to be estimated per device. Time division duplexing (TDD) offers a simplified CSI acquisition in a massive MIMO system with a relatively small number of devices. 
Assuming channel reciprocity, the channels are estimated at the BS relying on uplink training based on {\it orthogonal} pilot sequences and are then used for both uplink and downlink beamforming. Hence, the length/number of the pilot sequences scales with the number of devices and not the number of BS antennas. In crowd and mMTC scenarios, the channel estimation faces fundamental limits because of two specific features. 
\emph{First}, the number of pilot sequences is limited by the dimensionality of the coherence interval of the channel, during which the channel can be considered constant and flat. When 
the crowd of device becomes very large, the orthogonal pilots are in severe shortage and the allocation policy of the pilot sequences becomes a central question. 
\emph{Second}, the intermittence of the data traffic also plays a significant role. In mMTC, each device sends data to the BS sporadically, at random time instants. Since not all devices are active simultaneously, the limited set of orthogonal pilot sequences can highly likely accommodate the subset of devices that are active at a particular instant. Hence, pilot allocation has rather to adapt and scale with the traffic activity pattern and not equal the actual number of devices present in the system. A natural choice is to decentralize pilot access to the devices and make it random, which leads to pilot collision, also known as \emph{pilot contamination}. 


While massive MIMO is a fairly mature research topic \cite{Larsson2014a,Huh2012a,Ngo2013a,Bjornson2016b,Hoydis2013a,Bjornson2016a}, the existing results on uplink capacity analysis in the literature \cite{Ngo2013a} assumes  a pre-defined pilot allocation and 
full data buffers at the devices. Those assumptions are not applicable to the case we study here. 
A joint pilot and data protocol based on random uplink access in a massive MIMO system 
has been recently considered in \cite{SDP2014,Sorensen2016a}, with the use of coded access and successive interference cancellation (SIC).
In~\cite{Bjornson2016e,Bjornson2017a}, a set of dedicated pilots is temporarily allocated to the active devices, which in turn apply random access to the pilot sequences in the set. This leads to collisions, which are resolved by giving priority to devices with good channels, while the data transmission is carried out without pilot collisions.
Both schemes \cite{SDP2014,Sorensen2016a,Bjornson2016e,Bjornson2017a} rely on channel hardening. 
In \cite{Sanguinetti2016a}, timing variations between the devices are used to devise a random access scheme for initial access, synchronization, and channel estimation. Some preliminary results on the effect of intermittent device activity can be found in \cite{BjornsonGC2015}. 

{In the approach proposed in this paper, transmission is organized into transmission frames consisting of multiple transmission slots. 
When an mMTC device has a data codeword to transmit,  it divides the codeword into multiple parts and, within a transmission frame, 
transmits  the codeword parts in multiple transmission slots.}
The channels are estimated from uplink pilots every time the  device transmits. 
In each time slot, each active device selects pseudo-randomly a pilot from a predetermined pilot codebook and, during the rest of the slot, it sends a part of the data codeword. Hence, a device performs \emph{pilot-hopping} over multiple slots and the hopping pattern can be used to identify the device and appropriately merge and decode the parts of its codeword at the BS. This approach is suitable to harvest the large array gain for delay-tolerant mMTC and 
low-power devices.  {It is important to note that our access scheme relies on non-orthogonal pilot-hopping patterns, each sequence consisting of orthogonal pilots. Another option would be to build hopping patterns based on non-orthogonal pilots and our scheme can be readily extended to that case.} 

{
When the number of transmission slots is sufficiently large, the pilot collision events get averaged, so that 
we can define a maximal achievable uplink sum rate.}
We determine a lower bound on the  sum rate that is tight thanks to channel hardening and when the total number
of terminals is large. This bound is our main metric for
performance assessment.
Furthermore, in order to optimize the system performance, the devices are allowed to transmit  with a certain \emph{activation probability} $p_a$ which is subject to optimization, along with the number of training sequences. 
As the numerical evaluation of our main performance bound can be cumbersome, optimization is based on alternative bounds, tailored to account for rather large variations in channel energy among the devices. 
We provide a performance analysis for different asymptotic cases. In particular, 
when the number of antennas $M$ and the duration of transmission slot $\tau_u$ are of the same order, the sum rate  scales of 
$\sqrt{M \tau_u}$. Heuristic solutions giving robust performance results indicate that one third of the transmission slot should be devoted to training while the average number of active devices should be at the order of $\sqrt{M \tau_u}$.
Compared to~\cite{deCarvalho2016a}, we provide performance analysis and solutions for the case where the channel energy varies across the devices.

The paper is organized as follows. 
Section~\ref{sec:sysmodel} describes the random access model, the communication protocol, and the channel model. 
Section~\ref{sec:Bound1} presents the main performance bound, Section~\ref{sec:Bounds2} the lower bounds used for optimization and  
Section~\ref{sec:Bounds3} the heuristic solutions. 
In Section~\ref{sec:Scale}, scaling laws are presented as a function of the number of antennas at the base station and the duration of a transmission slot. Section~\ref{sec:numerical} contains numerical evaluations and is followed by the conclusions.

\section{System Model}
\label{sec:sysmodel}

\subsection{Random Access Model}
\label{sec:RA}

We consider the uplink (UL) of a single-cell multi-device massive MIMO system
 with random access from a large set of intermittently active devices. 
The justification for a single-cell model is found in Section~\ref{sec:MultiCell}.
The BS is equipped with $M$ antennas and  serves a maximal number of $K$ devices.
%

A total number of $\tau_p$ orthogonal sequences are available, denoted as 
 $\{ \boldsymbol{\Phi}_1 , \boldsymbol{\Phi}_2 , \dots, \boldsymbol{\Phi}_{\tau_p} \}$, where each sequence is $\tau_p$ symbols long. 
The duration of a pilot sequence is smaller than the channel coherence interval. 
Moreover, $K \gg \tau_p$ such that the BS cannot allocate dedicated pilots to particular devices.
{
The system is assumed to be synchronized to maintain pilot sequence orthogonality. }
{
It is important to note that an OFDM system can tolerate timing mismatches provided that the cyclic prefix is larger than the channel delay spread and the maximal timing mismatch. Such conditions could be satisfied in small enough cells in LTE as the normal cyclic prefix is designed for a 750 m cell radius and the extended cyclic prefix corresponds to a 2.5 km cell radius. 
}


A UL transmission frame, see Figure~\ref{fig:Model3}, is divided into transmission slots of duration $\tau_u$, where 
$\tau_u$ is less or equal to the channel coherence interval. 
Due to the random access, the BS does not know a priori for a given slot which devices transmit or which pilots are activated. Instead, each device is associated to a unique, predefined pseudo-random \emph{pilot-hopping pattern} such that an active device selects the pilot for a given transmission slot according to this pattern. 
The BS  knows in advance the pilot-hopping patterns of all potential transmitters, such that it can buffer the information from different slots, run a correlation decoder across the slots and find out which pilot-hopping patterns have been activated. 
{When a device has data to send, it encodes it into codewords, so that one codeword is sent within a given transmission frame. 
Each codeword is divided into multiple parts and sent over multiple time slots within the transmission frame. 
The number of codeword parts  is equal to the number of pilot sequences in the pilot-hopping identifying patterns. 
Transmission from the $K$ devices is sporadic and a device is active within a given transmission frame with activation probability $p_a$, independently from the other devices.  }
{An active device  transmits in all transmission slots of a transmission frame. Within a transmission frame, a device is active only when it has data to transmit.}


In each active UL transmission slot, the pilot phase is followed by a \emph{data phase}, i.e., transmission of a part of a codeword. Collisions can thus happen in the \emph{pilot domain}, i.e., among contending devices that send the same pilot sequence to the  BS. 
Pilot collision is in fact similar to pilot contamination in conventional massive MIMO and induces interference in the data domain from the set of contaminators. 

For an asymptotically large number of time slots, the whole codeword is affected by an asymptotically large number of channel fading realizations, pilot collision-induced interference and other interference events. Relying on the ergodicity of such a process, the system performance can be characterized as in Section~\ref{sec:Bound1}. Both the activation probability $p_a$ and the number $\tau_p$ of pilots per slot are optimized as elaborated in Section~\ref{sec:Bounds2}.


We model the process of constructing of a pseudo-random hopping sequence by having each active device in each slot select randomly one of the $\tau_p$ sequences. Hence, 
%
the probability of having  $K_a$ active devices  within a total of $K$  devices that become independently active with probability $p_a$ is
\begin{equation}
p(K_a) 
= \left(\!\! \begin{array}{c} K \\ K_a \end{array}\!\! \right) p_a^{K_a } (1-p_a)^{K_a  - K}.
\label{eq:pKa}
\end{equation}
The average number of active devices is $p_a K$ and its variance is $p_a K (1- p_a)$.
Given a number of $K_a$ active devices and considering one given device $0$ among the $K_a$ active devices,  
the probability of having $c$ colliders to device $0$ can be expressed as the probability that $c$ devices select the same sequence as device $0$ in a population of $K_a - 1$, i.e.,
\begin{equation}
p(c|K_a) 
= \left(\!\! \begin{array}{c} K_a \! - \! 1 \\ c \end{array}\!\! \right) 
\left(\frac{1}{\tau_p}\right)^{c} \left(1-\frac{1}{\tau_p} \right)^{K_a-c-1}.
\label{eq:pcKa}
\end{equation}

\begin{figure}[!t]
\centering
\includegraphics[width=8.3cm]{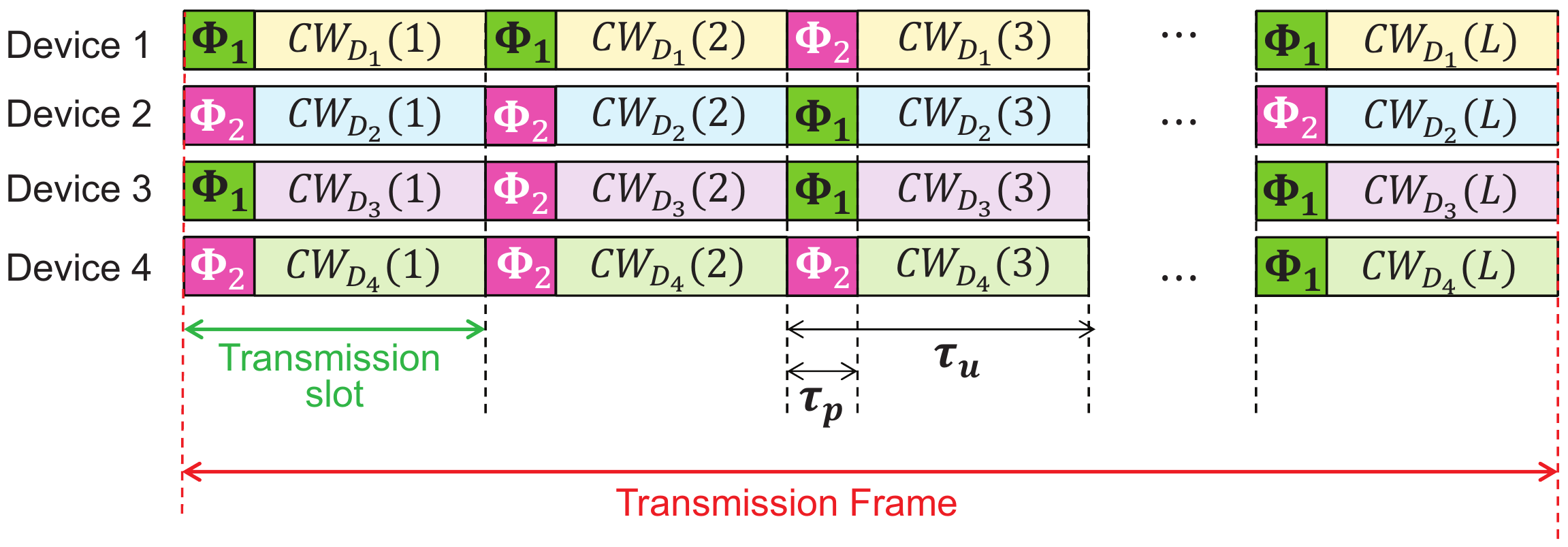} 
\caption{Illustration of the transmission frame. In this example, four devices $\{D_1,D_2, D_3,D_4\}$ and two mutually orthogonal pilot sequences  $\{\boldsymbol{\Phi}_1, \boldsymbol{\Phi}_2\}$ are considered. Transmission of a codeword is done over multiple channel fades, which enables averaging over noise, channel fades, and pilot collision events.}
\label{fig:Model3}
\end{figure}

\subsection{Channel Model}
\label{sec:Models}

A block fading model is adopted where a channel realization is constant across a transmission slot duration and changes independently from slot to slot.
A transmission slot is defined as a time-frequency interval that is smaller than the coherence time and coherence bandwidth of the channel.  
{In our model, the received signal is divided by the standard deviation of the receive additive noise, so that the normalized $M \times 1$ receive noise vector $\bfn$ at the BS is modelled as 
$\bfn \sim {\cal C \cal N}(\mathbf{0}, \mathbf{I}_M)$.
The channel response between the BS and device $j$ is described by an $M \times 1$ channel vector $\bfg_j$. The channel realizations are modeled as circularly symmetric complex Gaussian distributed:
$\bfg_j \sim {\cal C \cal N}(\mathbf{0}, \beta_j \mathbf{I}_M)$. 
The variance
${\beta_j}$ reflects the path loss, shadowing, transmit power (relative to the received noise power). 
One given parameter ${\beta_j}$ is assumed invariant during the whole transmission frame. }

{
The devices transmit with different powers: device $j$ transmits with power $p_j$ during both the training and data phase. 
We recall that the parameter $\beta_j$ incorporates the transmit power $p_j$.
Three models are considered for the distribution of $\beta_j$. 
In the  first two models, power control (performed during the initial access phase \cite{Dahlman:2011,Bjornson2017a}) is assumed, meaning that the transmit power from each device is adjusted so that all the $\beta_j$'s are equal.}
{
\begin{description}
\item[Model 1:] $\beta_j = \overline{\delta}_{1} (1+ v)$, where $v$ is uniformly distributed between $-\alpha$ and    
$\alpha$, with  $0 \leq \alpha \leq 1$. 
It is assumed that an imperfect power control is performed so that there is an error around a nominal value $\overline{\delta}_{1}$. 
This model was chosen as large variations of  $\alpha$ do not change significantly the optimization in Section~\ref{sec:Scale}, while the following two models do. 
\item[Model 2:] $\beta_j =  \overline{\delta}_{2} 10^{v/10}$ or 
$10\log_{10}(\overline{\beta}_{j}) =  10\log_{10}(\overline{\delta}_{2}) + v$,
 where $v$ is modelled as a Gaussian random variable with variance $\sigma_v^2$.
In this model, the parameter $\beta_j$ follows a log-normal distribution, i.e. the conventional distribution for shadowing. Power control  is performed so that the mean of the logarithm of $\beta_j$ is the same for all devices. 
\item[Model 3:] $\beta_j =\overline{\delta}_{3} (d/d_0)^{-\alpha_p}$, $d = d_0 (1+v)$, where $v$ is  uniformly distributed between  $-\alpha$ and $\alpha$, with  $0 \leq \alpha \leq 1$.   
Parameter $d_0$ is set to $500$ meters and $\alpha_p=3.76$.
This model accounts for the path loss when  devices are  distributed uniformly around a nominal distance from the cell center.
\end{description}
In the simulations, $\overline{\delta}_{1}=\overline{\delta}_{2}=\overline{\delta}_{3}$ and 
$\overline{\delta}_{j}$ is set to 10\,dB. 
As will be made clear in Section~\ref{sec:Bounds2}, 
the final results do not significantly depend on the values selected for $\overline{\delta}_1$, $\overline{\delta}_2$, $\overline{\delta}_3$, 
provided  we are in interference-limited scenarios. 
Furthermore, we  introduce the notations $\E [\beta_j] =\overline{\beta}$, $\E [\beta^2_j] =\overline{\beta^2}$,  $\E [\beta^4_j] =\overline{\beta^4}$. 
Note that $\overline{\beta}$ is equal to $\overline{\delta}_1$ for Model~1 but not to $\overline{\delta}_2$ or 
$\overline{\delta}_3$ for Model~2 and Model~3. }

We use  $(\cdot)^*$,  $(\cdot)^T$, $(\cdot)^H$, $\E[\cdot]$  to denote complex conjugation, transpose, Hermitian transpose, and the expected value of a random variable, respectively.

{
\subsection{Training Phase}
During the training phase, the active devices send the normalized pilot training sequences 
$\boldsymbol{\Phi}_k$, $\boldsymbol{\Phi}_k^H \boldsymbol{\Phi}_k=1$, of length $\tau_p$. 
These sequences are orthogonal: $\boldsymbol{\Phi}_k^H \boldsymbol{\Phi}_{k'}=0$, for $k\neq k'$.
The $M \times \tau_p$ matrix of the received pilot signals at the BS is
\begin{equation}
 \mathbf{Y}_p = \sqrt{\tau_p  } \sum_{k=0}^{K_a} \mathbf{g}_k \boldsymbol{\Phi}_k^T + \mathbf{N}_p ,
 \label{eq:RecSig}
\end{equation}
where $K_a$ is the number of active devices in a given slot and $\mathbf{N}_p$ is an $M \times \tau_p$ matrix whose columns are the noise vectors during the training phase. \\
\indent
One given device, device $0$, is affected by a set of contaminators denoted as $\cC_0$, which use the same pilot sequence that we generically denote as 
$\boldsymbol{\Phi}_{[0]}$, i.e. $\boldsymbol{\Phi}_{[0]} = \boldsymbol{\Phi}_{j}, j\in \{1, \cdots, \tau_p\}$.
}

{
\subsection{Multi-Cell Operation}
\label{sec:MultiCell}
The main line of work on massive MIMO, starting from \cite{Marzetta2010a}, has assumed that all devices in a given cell use orthogonal pilots and analyzed the system performance under inter-cell pilot collisions due to pilot reuse across cells. 
Our main assumption is that intra-cell pilot collision happens because the crowd of devices calls for random pilot access. 
This assumption is by itself a departure from the conventional multi-cell interference models. 
Note that \cite{Caire13b} considers inter-cell pilot contamination in the sense that it provides a pilot reuse plan 
within a cell so that users with very different spatial characteristics use the same pilot, hence suppressing pilot contamination. 
In contrast, our system is uncoordinated with no specific assumptions of the channel properties, so the intra-cell pilot contamination can be strong.
Concerning the inter-cell pilot collision, there are two possibilities. 
\begin{itemize}
\item One possibility is to assign orthogonal pilot sets to each cell where the sets are also orthogonal across neighboring cells, so that one can neglect the inter-cell pilot collision. This protocol has the disadvantage of decreasing the number of available orthogonal pilot sequences per cell for a given pilot sequence length, as the orthogonal sequences are distributed across cells. The advantage is that pilot contamination can be better controlled  when it is confined to a cell as opposed to multiple cells. 
This model is equivalent to the single-cell model in Section~\ref{sec:RA} as far as inter-cell pilot contamination is concerned. In data transmission mode, inter-cell interference is present. 
\item  The other possibility is to have the whole set of orthogonal sequences available in each cell, resulting in inter-cell pilot collision. 
Assuming the same activation probability in all the cells and a uniform spatial coverage, the difference with a single-cell model is:
\begin{itemize}
\item The distribution of channel gains $\{\beta_j\}$. 
Model 3 can still be applied to multiple cells as it is based on distance to the BS of interest. However, it is not the case for Model~1 and Model~2 as the underlying parameters take as reference the cell the device belongs to. The distribution would be different for the devices in the cell and the devices in the interfering cells. 
\item The sum rate involves only the devices in the cell of interest. 
\end{itemize}
\end{itemize}
In this paper, we focus on a single-cell model. 
As the random access protocol presented in this paper is new, an initial analysis based on this simplified model is necessary to understand the main mechanism behind the protocol. The study of a multi-cell system as presented above is left for future work.
}

\section{Random Access Protocol}
\label{sec:ProtocolDescription}

{
The  pilot and data random access protocol at the BS is outlined in Section~\ref{sec:ProtDesc}. It relies on the channel estimation, device identification and decoding procedure presented in Sections~\ref{sec:ChanEst}, \ref{sec:DevID} and \ref{sec:dec}. 
\subsection{Protocol Description at Receiver} \label{sec:ProtDesc}
The following procedure is considered  for the  transmission of one codeword per device within a transmission frame.  
The notations are defined in Section~\ref{sec:ChanEst}.
\\[2mm]
\noindent  {\bf Step 1: Multiple slot processing.} \\
For one transmission slot, the following processing steps are performed:
\begin{enumerate}
\item {\bf Pilot sequence detection:} The BS detects which pilot sequences  are in use. 
This is done by correlating the received signal $\mathbf{Y}_p$ in  \eqref{eq:RecSig} with each sequence $\boldsymbol{\Phi}_j$ that is available: $\mathbf{Y}_p \boldsymbol{\Phi}_j^*$, $j=1,\ldots, \tau_p$. Pilot $j$ is detected if $\|\mathbf{Y}_p \boldsymbol{\Phi}_j^* \|$, $j=1,\ldots, \tau_p$ is larger than a pre-determined threshold and nondetected otherwise. {The pilot detection outcomes are buffered in order to be utilized for device activity detection, see Step 2.}
\item {\bf Channel estimation:} 
For each pilot sequence detected, a  scaled version of the  MMSE channel estimate is determined. This estimate is the same for all members of a contamination set. 
\item {\bf Multi-antenna processing:}  For each pilot sequence detected, Maximum Ratio Combining  (MRC) is applied to the data symbols  in the slot and the MRC  output is buffered along with its associated pilot index. 
\item  {\bf End of transmission:} The process is repeated until all active devices have stopped transmitting their codeword. End of transmission is jointly detected with the active device identification. 
\end{enumerate}
\noindent   {\bf Step 2: Active device identification.}  
From the pilot-hopping patterns detected across the multiple slots at Step 1, the
transmitting devices are detected.
\\[1mm]
\noindent  {\bf Step 3: Decoding.}  
Based on the identified pilot patterns, the BS identifies which MRC outputs to combine to decode the data of each transmitting device.  
For each identified device $0$, data decoding assumes  the knowledge of ${\beta_0}$. 
}

\subsection{MMSE Channel Estimation}
\label{sec:ChanEst}

The MMSE estimation of $\mathbf{g}_{[0]}$ proceeds first by correlating $ \mathbf{Y}_p$ with the pilot sequence $\boldsymbol{\Phi}_0$ to obtain
\begin{equation}
 \mathbf{y}_p = \mathbf{Y}_p  \boldsymbol{\Phi}_{[0]}^* = \sqrt{\tau_p  } \sum_{j \in \cC_0} \mathbf{g}_j  + \mathbf{N}_p  \boldsymbol{\Phi}_{[0]}^*.
 \label{eq:ts}
\end{equation}
The MMSE estimate of $\mathbf{g}_0$  is then \cite{Kay1993a}
\begin{equation}
\hat{\mathbf{g}}_{0}  = 
\left( \E [ \bfy_p \bfy_p^H ] \right)^{-1}{ \E [\bfh_0  \bfy_p^H ] } \; \mathbf{y}_p  = 
\frac{\sigma_{\bfh_0 \mathbf{y}_p}^2}{\sigma_{ \mathbf{y}_p \mathbf{y}_p}^2}\mathbf{y}_p ,
\label{eq:mmse}
\end{equation}
where $\E [\bfy_p \bfy_p^H] = \sigma^2_{\bfy_p \bfy_p}  \mathbf{I}_M$ and 
$\E [ \bfh_0  \bfy_p^H ]  =  \sigma^2_{\bfh_0  \bfy_p}  \mathbf{I}_M$, with
\begin{eqnarray}
\sigma^2_{\bfy_p \bfy_p} & = &{\tau_p  } \beta_0 + {\tau_p  }\sum_{j \in {\cal C}_0}^{ }  \beta_j  + 1, \\
 \sigma^2_{\bfh_0  \bfy_p} & = & \sqrt{\tau_p  } \beta_0.
\end{eqnarray}

The following results will be used in the computation of the performance bound in Section~\ref{sec:PerfBounds}. For $j \in {\cal C}_0$
\begin{equation}
\hat{\mathbf{g}}_j = \frac{\beta_j}{\beta_0} \hat{\mathbf{g}}_0, \quad 
\E [ \hat{\mathbf{g}}_0 \hat{\mathbf{g}}_0^H ]= 
\sigma_{ \hat{\mathbf{g}}_0}^2 \mathbf{I}_M, \quad 
\sigma_{ \hat{\mathbf{g}}_0}^2 = 
\frac{\sigma_{\bfh_0 \mathbf{y}_p}^4 }{\sigma_{\mathbf{y}_p \mathbf{y}_p}^2}
\label{eq:mmse3}
\end{equation}
\begin{align}
\hspace{-6mm}
\E [ \bm{\varepsilon}_0 \bm{\varepsilon}_0^H ] &= \sigma_{ \hat{\mathbf{g}}_0}^2 \mathbf{I}_M, 
\nonumber
\\
\sigma_{ \hat{\mathbf{g}}_0}^2 &= \beta_0 - \frac{\sigma_{\bfh_0 \mathbf{y}_p}^4 }{\sigma_{\mathbf{y}_p \mathbf{y}_p}^2} = 
 \frac{ {\tau_p  }\sum_{j \in {\cal C}_0} \beta_j  + 1 }{\sigma_{\mathbf{y}_0 \mathbf{y}_0}^2} .
\label{eq:mmse2}
\end{align}
where $\bm{\varepsilon}_0 = \hat{\mathbf{g}}_0 - \mathbf{g}_0$. 

We define the quantity $\hat{\mathbf{g}}_{0}'$ as: 
\begin{equation}
\hat{\mathbf{g}}_{0} =  \beta_0 \hat{\mathbf{g}}_{0}' , \quad 
\hat{\mathbf{g}}_{0}' =  \frac{\sqrt{\tau_p  } }{ {\tau_p  }\sum_{j \in \{0,  {\cal C}_0 \} }  \beta_j  + 1}
\label{eq:LS}
\end{equation}
We  use $\hat{\mathbf{g}}_{0}' $ as a basis for the MRC   operation at each transmission slot  (see Section~\ref{sec:ProtDesc}).   

{
The sum power $\sum_{j \in \{0, \cC_0\} } {\beta_j}$ can be estimated accurately as detailed in\cite{Bjornson2017a}: indeed, using channel hardening, as the number of antenna grows large, the term $\boldsymbol{\Phi}_j^T \mathbf{Y}_p^H   \mathbf{Y}_p \boldsymbol{\Phi}_j^*$ is asymptotically equal to $\sigma_{ \mathbf{y}_p \mathbf{y}_p}^2 = \tau_p \sum_{j \in \{0, \cC_0\} } {\beta_j} + 1$. }
Then, 
for each pilot sequence detected, the corresponding channel estimate $\hat{\mathbf{g}}_{0}' $ is computed using \eqref{eq:LS}.

{
\subsection{Active Device Identification}
\label{sec:DevID}
Each pseudorandom pilot-hopping pattern serves as a unique signature/identification for the device that applies that pattern. The device patterns are known a priori to the BS, such that the BS combines the pilot sequence detection outcomes from the slots that follow the pattern in order to detect the transmitting devices.
The pseudorandom pilot-hopping pattern detection follows the same principle as pseudorandom CDMA sequence detection. 
We assume that the number pilot hopping sequences is sufficiently large so that every user in the cell has a unique identifying sequence. 
}

\subsection{Decoding Procedure}
\label{sec:dec}

The MMSE channel estimate of device $0$ requires knowledge of  parameter $\beta_0$ which cannot be assigned to device $0$ until it is identified. 
To overcome this limitation, data decoding is performed as follows. At first, the BS obtains, for each transmission slot, the output of the MRC receiver based on the channel 
estimates $\hat{\mathbf{g}}_{0}' $ defined in \eqref{eq:LS} which are estimated at each transmission slot. 
A pilot sequence is linked to each of the MRC output at each transmission slot.  
MRC is attractive in massive MIMO due to its low computational complexity and asymptotic rejection of non-coherent multi-device interference and noise~\cite{Larsson2014a}. 
Next, the BS identifies a transmitting device through its identifying pseudo-random  sequence from which it can associate the relevant MRC outputs over the multiple transmission  slots. 
We assume that  the BS knows the associated $\beta_0$ parameter {(based on the initial access phase \cite{Dahlman:2011,Bjornson2017a})} and can estimate $\sum_{j \in \{0, {\cal C}_0\}}\beta_j $ 
based on equation (\ref{eq:ts}) (using channel hardening as in~\cite{Bjornson2017a}). 
The BS  performs soft combining of the MRC outputs scaled by $\beta_0$.


%

\begin{figure*}[t]
\begin{equation}
\resizebox{.9\hsize}{!}{$
\underline{\SINR}_1 ( \cC_0, K_a , \{\beta\} )= 
\frac{  \tau_p (M-1) \beta_0^2 }
{
  \tau_p (M-1) \sum_{j \in \mathcal{C}_0} \beta_j^2 
 + 
%
 \sum_{i \in  \{0, \cC_0  \} }  \beta_i (1 +  \tau_p \sum_{j \in {\cal C}_i}^{ }  \beta_j )
+ 
%
( 1+  \sum_{i \notin \{0, \cC_0 \}} \beta_i ) 
({1 +   \tau_p \sum_{i\in {\{0, {\cal C}_0} \}}^{ }  \beta_i } )}
$}
\label{eq:sinr1}
\end{equation}
\begin{eqnarray}
\resizebox{.75\hsize}{!}{
$
\underline{\SINR}_{2}  (c, K_a,\beta_0) = 
%
\frac{ \tau_p\left(M-1\right)  \beta_0^2 }{ 
%
 \tau_p  \left( M\!-1\right)   
\overline{\beta^2} 
c
%
+
%
  \beta_0  (1 +  \tau_p c\overline{\beta}  )  - c  \overline{\beta}^2  \tau_p 
+
%
   ( 1 + (K_a-1)\overline{\beta} )(   1  +  
 \beta_0    \tau_p   +  \tau_p c \overline{\beta}   ) 
%
}
$
}
\label{eq:sinr2}
\end{eqnarray}

\begin{eqnarray}
\resizebox{1\hsize}{!}{
$\underline{\SINR}_{3} (\beta_0) \!= \!
\frac{\tau_p\left(M-1\right)  \beta_0^2}{ 
%
\overline{\beta^2}
 \left( M\!-1 \right) 
 \left(  p_a K-1 \right) 
%
+
%
   \beta_0 (1 +  \overline{\beta}   \left(  p_a K-1 \right)  ) 
 -   \overline{\beta}^2  \left(  p_a K-1 \right) 
+
%
  ( 1 +  ( p_a K -1)\overline{\beta} )(  1 +  
 \beta_0  \tau_p     )  +
( p_a K -1)\overline{\beta}   +
 \overline{\beta}^2
 \left(  p_a^2 K (K-1) - (p_a K-1) \right)
%
}$
}
\label{eq:sinr3}
\end{eqnarray}
\hrule
\end{figure*}

\section{Uplink sum rate}
\label{sec:Bound1}

{
\subsection{Maximal Achievable Sum Rate and Main Assumptions}
\label{sec:MainAssump}
Within a transmission frame, the
codeword of a given device experiences all possible contamination
events from the $K_a$ active devices, provided that the number of
transmission slots  is sufficiently large. Likewise, for
an asymptotically large number of transmission slots, a device experiences 
an asymptotically large number of   additive noise  and fading realizations.
Under those asymptotic conditions, a maximal achievable rate per device
can be defined within each transmission frame.\\
Achieving this rate assumes the following features: 
\begin{enumerate}
\item The number of active devices within a transmission frame can be estimated at the BS.
\item The average channel energy $\beta_j$ per device $j$ can be estimated at the BS and at the device.
\item 
Under conditions 1) and 2), the BS can compute a maximal achievable rate per device. The BS broadcasts both
the channel energy  $\beta_j$ and the associated rate for each active device.
As the device itself knows its channel energy, it can associate
the corresponding rate.
\end{enumerate}
The different performance metrics presented in this paper correspond to the maximal achievable sum rate averaged over the activation probability of the devices. 
}

We present two types of bounds on the average maximal achievable sum rate. 
First, the main bound $\cR_1$ in (\ref{eq:cR1}) tightly approximates the maximal achievable sum rate and is used for performance assessment. 
It necessitates  Monte-Carlo simulations of the channel large scale fading realizations (i.e., $\beta$). 
Its numerical evaluation can be heavy as it involves the computation of expected values over binomial distributions with large number of devices $K$. 
A secondary  bound, $\cR_2$, does not require Monte-Carlo simulation but its tightness depends on the distribution of the parameters $\{\beta_j\}$. Other bounds  $\cR_3$ in \eqref{eq:cR3} and $\cR_a$ in \eqref{eq:cRa} are developed and serve as optimization tools. 
They are relatively loose, but analytically simple and follow the variations of the ergodic sum rate well.
Those bounds are the basis for the heuristic solutions presented in Section~\ref{sec:Bounds3}.

\subsection{Main Performance Bounds}
\label{sec:PerfBounds}

Here we present the main bound $\cR_1$  that is used for performance assessment. 
A secondary bound $\cR_2$ is also provided and mainly helps in establishing bound $\cR_3$ in Section~\ref{sec:Bounds2}.

\begin{theorem}[Main bound $\mathcal{R}_1$]
\label{theorem:first-theorem}
An approximation of the average maximal achievable sum rate (in bits per symbol) is
\begin{eqnarray}
\cR_1 = 
\!\!
 \sum_{K_a=1}^K  p(K_a) \, K_a  \sum_{c=0}^{K_a-1}   p(c|K_a) \;
\E_{\beta}   \left[  R_1 \left(\cC_0,  K_a , \{\beta\} \right)   \right] 
\label{eq:cR1}
\end{eqnarray}
where 
\begin{eqnarray}
 R_1  \left( \cC_0, K_a , \{\beta\}  \right)   =  \frac{\tau_u - \tau_p}{\tau_u}  \log_2 \big(1+ \underline{\SINR}_1( \cC_0, K_a , \{\beta\} ) \big) 
\label{eq:R1}
\end{eqnarray}
is a lower bound on the maximal achievable rate of device $0$ conditioned on a collider set with indices $\cC_0$ and $K_a$ active devices. The expectation is taken with respect to $\beta_j, j\in \{0, \cC_0\}$, denoted as  $\{\beta\}$. $\underline{\SINR}_1(\cC_0, K_a , \{\beta\})$ is given by (\ref{eq:sinr1}) at the top of the next page.

%
Recall that $p(c|K_a)$ is the probability of having $c$ colliders to device $0$ when $K_a$ devices are active and is given in Equation~\eqref{eq:pcKa}. 
$p(K_a)$ is the probability of having $K_a$ active devices out of $K$ and is given in Equation~\eqref{eq:pKa}. 
\end{theorem}
\begin{proof}
See  Appendix~\ref{sec:proof1}.
\end{proof}

The computation of the sum rate  involves averaging w.r.t. the specific realization of the $\beta$'s during the whole transmission frame and computing quantities related to all possible sets of different size among the $K$ devices 
(See  Appendix~\ref{sec:proof1}). 
As it is prohibitively complex, we instead compute an 
approximation \eqref{eq:cR1} where the expected value w.r.t. the discrete probability functions is
replaced by an expected value w.r.t. continuous probability functions. As $K$ becomes large, the approximation becomes very good. The results show that $\cR_1$ approaches a limit when $K \to \infty$. 

Furthermore, after a certain threshold on $K$,  $\cR_1$  does not increase significantly. This can be understood as 
the tail of the binomial distribution becomes very small for large values of $K$ and brings the associated expected value of  term
$K_a  \sum_{c=0}^{K_a-1}   p(c|K_a) \;
\E_{\beta}   \left[  R_1 \left( \{\beta\} , \cC_0|K_a \right)   \right]$ in \eqref{eq:cR1} to zero. 
This result also means that after this threshold, the individual rate per device diminishes. 

Under the assumptions described in Section~\ref{sec:MainAssump}, 
 device $0$ computes  
its transmission rate as
\begin{eqnarray}
\cR_0 = 
\!\!
\frac{1}{M} \sum_{K_a=1}^K  p(K_a) K_ a \,  \sum_{c=0}^{K_a-1}   p(c|K_a) \;
\E_{\{\beta\} \setminus \beta_0}   \left[  R_1 \left(\cC_0 , K_a  ,\{\beta\}  \right)   \right]
\label{eq:cR12}
\end{eqnarray}
where $\E_{\{\beta\} \setminus \beta_0}$ designates the expectation w.r.t. $\{\beta_k, k\neq 0\}$.

\subsection{Secondary Performance Bounds}

The next bound  $\cR_2$ is a secondary bound  that is used as a basis for the bounds in Section~\ref{sec:Bounds2}. 
Using  Jensen's inequality on $\log_2(1+1/x)$, a lower bound is obtained by taking the expected value of the denominator in (\ref{eq:sinr2}) w.r.t. all sets of contaminators to device $0$. 
The main advantage of $\cR_2$ compared with $\cR_1$ is that the equivalent SINR, $\underline{\SINR}_{2}(\beta_0)$, in (\ref{eq:sinr2})  can now be explicitly written as a function of the number of contaminators to device $0$ and number of active devices. 
%

\begin{theorem}[Secondary bound: $\mathcal{R}_2$]
A lower bound on  $\cR_{1}$  (in bits per symbols) is
\begin{eqnarray}
\cR_{2}=  
\sum_{K_a=1}^K  p(K_a) \, K_a\sum_{c=0}^{K_a-1}   p(c|K_a)   
\E_{\beta_0} [ R_{2}(c, K_a , \beta_0) ]
\label{eq:cR2}
\end{eqnarray}
where
\begin{eqnarray}
R_{2}(c, K_a , \beta_0 ) =  \frac{\tau_u - \tau_p}{\tau_u}  \log_2(1+ \underline{\SINR}_{2}(c, K_a , \beta_0))
\label{eq:R2}
\end{eqnarray} 
and $\underline{\SINR}_2(c, K_a , \beta_0)$ is shown in (\ref{eq:sinr2}) at top of the next page. 
$p(c|K_a)$ and   $p(K_a)$ are defined in~\eqref{eq:pKa} and \eqref{eq:pcKa}.  
\end{theorem}
In (\ref{eq:cR2}), we maintain the expected value w.r.t. $\beta_0$  outside of the $\log$ term, 
as its impact is not easily captured through simplifying tools like Jensen's inequality, if  
the $\beta$'s have a large variance across the device population. 
Note that $\underline{\SINR}_{2}(c, K_a , \beta_0))$ is a fraction depending on $\beta_0$ and $\beta_0^2$ and its integral w.r.t. 
$\beta_0$ can be written as a closed form expression, so that  
$\E_{\beta_0} [ \log_2(1+ \underline{\SINR}_{2}(c, K_a , \beta_0)) ]$ can be written in closed form for Model 1, i.e., a uniform distribution of $\beta$. 

\section{Bounds for Optimization and Performance Analysis}
\label{sec:Bounds2}

We now determine the values of the activation probability $p_a$ and the pilot sequence length $\tau_p$ that maximize the performance bound $\cR_1$. As the numerical evaluation of $\cR_1$ is heavy for large values of $K$, we use 
two lower  bounds on $\cR_1$ that are the main bounds used for optimization as well as performance analysis. Such an optimization assumes the knowledge of the distribution of the  $\beta$'s. 
As $p_a K$ is in fact the optimization parameter, the knowledge of the total number of devices $K$ is not necessary for the optimization itself, but it is necessary to decide the reliable rate for each device. 

We mostly treat the case where $\tau_p$ is not linked to an actual channel, meaning that we do not integrate in the optimization problem that the constraint $\tau_p$ should be smaller than the coherence time or is limited by the transmit power of the devices. The main results from the latter case are mentioned in Appendix~\ref{sec:App2}.
Furthermore, as mentioned in Section~\ref{sec:Models}, we assume that a device always has data to transmit at activation time. 
To summarize, the optimization problem is: 
\begin{equation}
\max_{0\leq \tau_p \leq \tau_u , 0 \leq p_a \leq 1}
\cR_o(\tau_p, p_a).
\label{eq:optG}
\end{equation}
$\cR_o$ denotes generically the different cost functions considered in Sections~\ref{sec:Bound1},~\ref{sec:Bounds2} and~\ref{sec:Bounds3}.

In the next bound $\cR_3$, the expectation is taken in the denominator of (\ref{eq:sinr2}) w.r.t.~the distribution of the number of contaminators and the number of active devices.
The bound $\cR_3$ is relatively loose as compared to $\cR_1$  and $\cR_2$, since it averages over the number of colliders and active devices in the interference variances.
However, $\cR_3$  follows very well the  variations of $\cR_1$ and provides very good optimization results provided that the variance of the parameter $\beta$ is not too large. 
This aspect is highlighted in the numerical results section.

\begin{theorem}[Bound for optimization: $\mathcal{R}_3$]
A lower bound on $\cR_1$  (in bits per symbol) is
\label{thm:cR3}
\begin{eqnarray}
\cR_{3} = \frac{\tau_u - \tau_p}{\tau_u}  \,  p_a K \, \E_{\beta_0} [ \log_2(1+ \underline{\SINR}_{3}(\beta_0)) ]
\label{eq:cR3}
\end{eqnarray}
where $\underline{\SINR}_3(\beta_0)$ is given in \eqref{eq:sinr3} at the top of the next page.
\end{theorem}
\begin{proof}
In the denominator of $\underline{\SINR}_2$, we take the expected value w.r.t. the probability mass of the binomial distribution $p(c) = p(K_a)p(c|K_a)$. 
More specifically, we take first the expected value of $c$ conditioned on a number of active devices $K_a$. It is the average number of contaminators to one given device and is equal to $(K_a -1 )/\tau_p$. Then, we take the expected value w.r.t. $K_a$, i.e., the average number of active devices out of $K$ devices which is equal to $p_a K$. Hence, 
$\E[c]= (p_a K-1)/\tau_p$. 
We obtain a lower bound on the ergodic  rate of a single device. Next, we multiply by the average number of devices $p_a K$.
\end{proof}

From (\ref{eq:sinr3}), a simpler bound can be deduced given asymptotic conditions on the parameters for which we keep the equivalent expression of $\underline{\SINR}_{3}(\beta_0)$. 
\begin{namedtheorem}[Asymptotic bound: $\mathcal{R}_a$]
If $K\gg 1$, $\tau_u \gg 1$,  $M \gg 1$, $\tau_p \gg 1$ and $p_a K \gg 1$, 
a lower bound on the ergodic sum rate is
\begin{eqnarray}
\cR_{a} = p_a K \frac{\tau_u - \tau_p}{\tau_u} E_{\beta_0}\!\left[ \log_2(1+ \underline{\SINR}_{a}(\beta_0)) \right],
\label{eq:cRa}
\end{eqnarray}
\begin{eqnarray}  
\underline{\SINR}_{a}(\beta_0) =     
         \frac{M \tau_p \beta_0^2 }{
        \overline{\beta^2}  M   p_a K +
                \overline{\beta}^2 p_a^2 K^2     +
         \overline{\beta} \beta_0 p_a K  \tau_p }.
\label{eq:sinra}
\end{eqnarray}
\end{namedtheorem}

\vspace{2mm}
The expression of $\underline{\SINR}_{a}(\beta_0)$ allows for an easy identification of the different terms and their effect. For that purpose, 
we rewrite $\underline{\SINR}_{a}(\beta_0)$  as
\begin{eqnarray}  
\frac{1}{\underline{\SINR}_{a}(\beta_0)}=     
        \frac{\overline{\beta^2}    p_a  K} { \tau_p \beta_0^2} +
                \frac{\overline{\beta}^2 p_a^2 K^2} {M \tau_p \beta_0^2}     +
          \frac{\overline{\beta} \beta_0 p_a K   } {M \beta_0^2} .
\label{eq:sinra2}
\end{eqnarray}
The first term in the sum comes from interference induced by pilot collision. 
The second and third terms come from the residual multi-device interference left by MRC.

We illustrate the variations of $\cR_a$ in Figure~\ref{fig:Ra_2D} as a function of $\tau_p$ and $p_a K$, for $M=\tau_u=100$ and no variations across the $\beta$'s (or perfect power control). 
For fixed $p_a K$, 
the sum rate first increases as $\tau_p$ increases  since the probability of pilot collision decreases. Hence, the first and second interference terms decrease. 
Then, the rate penalty $\frac{\tau_u - \tau_p}{\tau_u}$ coming from longer pilot sequences results in a sum rate decrease. 

For a fixed number of pilot sequences $\tau_p$, the sum rate first increases through the pre-log term: the number of devices increases faster than the SINR decreases from increased interference. 
When the decrease of the SINR becomes too large, the sum rate decreases. 
This decrease comes from the multi-device interference term $\frac{\overline{\beta}^2 p_a^2 K^2} {M \tau_p \beta_0^2}$.  
This last observation will be the basis of heuristic solution \eqref{eq:heur1} in Section~\ref{sec:Bounds3}.
\begin{figure}[t]
\centering
\includegraphics[width=8cm]{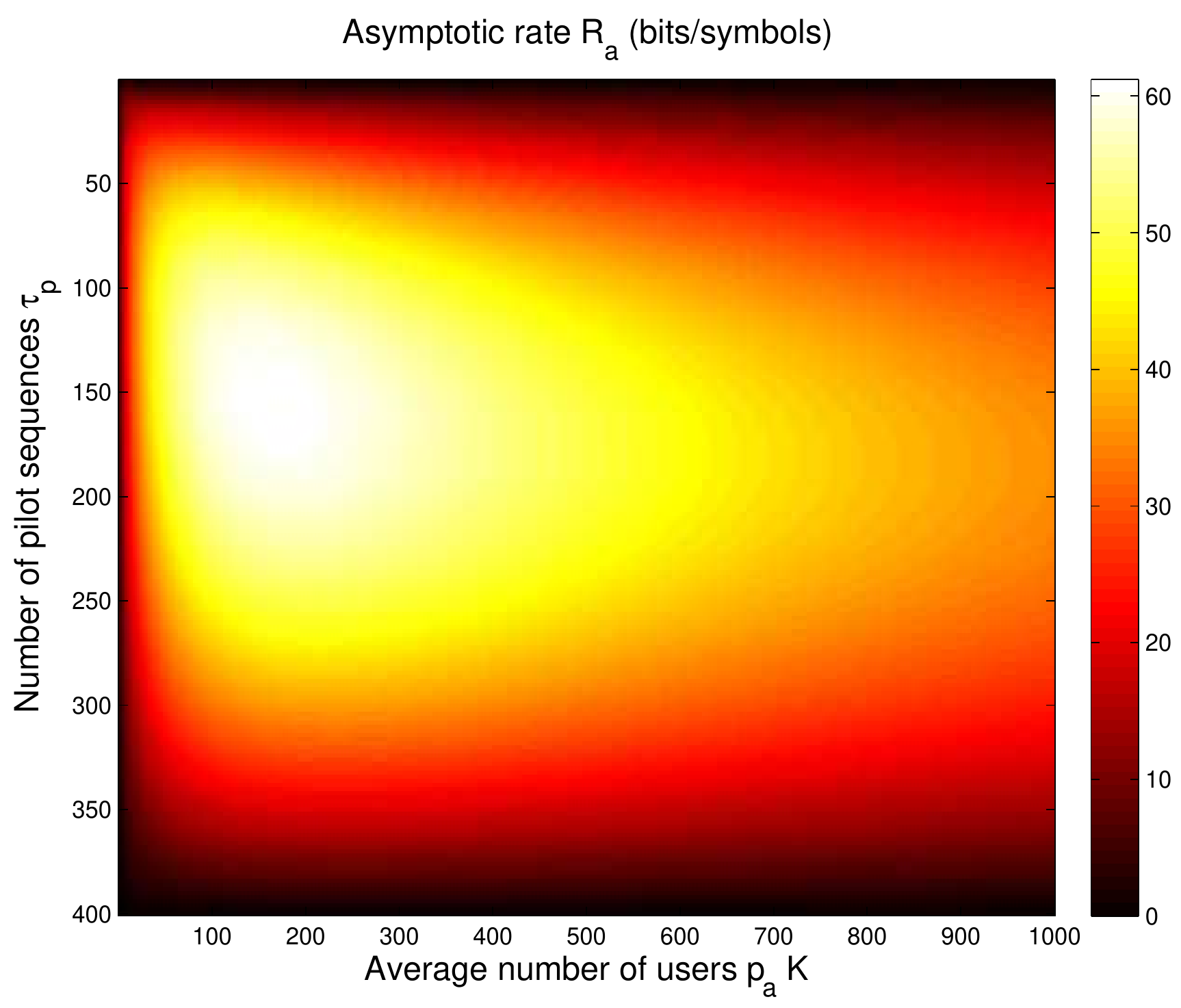}
\caption{$\cR_a$  as a function of $\tau_p$ and $p_a K$ for $M=\tau_u=400$.}
\label{fig:Ra_2D}
\end{figure}

\section{Heuristic Solutions}
\label{sec:Bounds3}

In this section, we present heuristic solutions to our optimization problem. 
To simplify the notations, we employ the superscript $()^o$ to denote the solutions of the optimization problem regardless of the optimization cost function.

\subsection{Small variations of $\{\beta\}$ }

Consider Model 1 and Model 3 in Section~\ref{sec:Models} with $\alpha \ll 1$. 
A Taylor approximation of the $\log$ term in \eqref{eq:cRa} is of the form 
$\log_2\left( 1+ \underline{\SINR}_{a}( {\beta}_0 ) \right)
=
\log_2(1+ \underline{\SINR}_{a}( \overline{\beta} ) ) + c_1 \alpha + O( \alpha^2)$, 
where $c_1$ is a scalar involving all the parameters except $\alpha$. 
Therefore,  $E_{\beta_0} [ \log_2(1+ \underline{\SINR}_{a}( {\beta}_0 ) ] = \log_2(1+ \underline{\SINR}_{a}( \overline{\beta} ) ) + O(\alpha^3)$. The same kind of derivation also holds for Model 3. 
We note that the approximation error is of order $\alpha^3$, which explains why the approximation is quite robust to values of $\alpha$ that are not strictly speaking ``small'', as seen in the numerical results. 
For Model 1, we found that the performance and optimization are quite insensitive to large values of $\alpha$. For Model 3, a value of $\alpha=0.25$ results in optimized performance that is close to the ones with perfect power control ($\alpha=0$), while, for Model~2, this value is $\alpha=0.1$.  

For small variance in $\beta$'s, the approximate cost function is
\begin{eqnarray}
\cR_{a} \approx p_a K \frac{\tau_u - \tau_p}{\tau_u} \log_2(1+ \underline{\SINR}_{a}(\overline{\beta})).
\label{eq:cRa_approx}
\end{eqnarray}

In Section~\ref{sec:Bounds2}, we noted that the variations of $\cR_a$   depend on the term in $p_a^2 K^2$ in the denominator of $\underline{\SINR}_a$. 
 The heuristic solution presented next is obtained based on \eqref{eq:cRa_approx} where only the term in $p_a^2 K^2$  is kept.
 
\begin{namedtheorem}[Heuristic Solution 1]
\label{thm:heuristic1}
A heuristic solution to the optimization of parameters  $\tau_p$ and $p_a$ maximizes the modified cost function
\begin{eqnarray}
\cR_{h}^{[0]} = p_a K \frac{\tau_u - \tau_p}{\tau_u}
\log_2 \left (1+  \frac{M \tau_p }{
                p_a^2 K^2     } \right).
\end{eqnarray}
The solution is 
\begin{eqnarray}   
 {\tau_p^o}= \frac{\tau_u}{3}\quad \quad
 p_a^o K = \frac{\sqrt{\tau_u M}}{\sqrt{3 s_0}}
 \label{eq:heur1}
\end{eqnarray}
where $s_0\approx 3.92$ is the solution of $\log(1+x) =2 \frac{x}{1+x}$. 
\end{namedtheorem}
\begin{proof}
We look for the expression of $\tau_p$ and $p_a K$ maximizing the following rate function: 
\vspace{-2mm}    
\begin{eqnarray}  
R^h =  p_a K (\tau_u - \tau_p)
 \log_2(1+X),  \quad
X =   \frac{M \tau_p }{
                \overline{\beta}^2 \overline{\beta^{-2}} p_a^2 K^2    }.
\end{eqnarray}

\vspace{-2mm}
\noindent 
The partial derivative of $R^h$ are
\vspace{-2mm}
\begin{eqnarray}     
\left\{ 
\begin{array}{lcl}
\frac{\partial R^h}{\partial  \tau_p} & =&
-  
\log_2(1+X) 
+
(\tau_u - \tau_p)
\frac{\partial X}{\partial  \tau_p}
\frac{1}{1+X},  \\
\frac{\partial R^h}{\partial p_a} & =& \log_2(1+X) 
+ p_a 
\frac{\partial X}{\partial  p_a}
\frac{1}{1+X}.
\end{array}
\right.
\end{eqnarray}

\vspace{-1mm}
\noindent
Noting that 
$
p_a  \frac{\partial X}{\partial  p_a} = 
-2 {X} $ and 
$
 \frac{\partial X}{\partial \tau_p} = 
\frac{1}{\tau_p} {X}
$, we obtain

\vspace{-3mm}
\begin{eqnarray}     
\left\{ 
\begin{array}{l}
1+\frac{1+X}{X} \log(1+X) =
 \frac{\tau_u}{\tau_p} , \\
\frac{1+X}{X}   \log(1+X) 
=2. \\
\end{array}
\right.
\end{eqnarray}
\vspace{-2mm}
From those equations, we obtain (\ref{eq:heur1}).
\end{proof}
An advantage of \eqref{eq:cRa_approx} is that it does not depend on the statistics of the $\beta$'s, such that they need not to be known.  


\subsection{Large variations of $\{\beta\}$}

For large values of $\alpha$,  the performance of the solution previously described degrades. We observe that the optimal solution involves a larger average number of active devices $p_a K$ while $\tau_u/3$ remains a good approximation for $\tau_p$. 
In the following heuristic solutions, we set $\tau_p^o = \tau_u/3$, while $p_a K$ is optimized based on $\cR_a$ or heuristic $\cR_h^{[0]}$. The superscript $1D$ below indicates that the optimization is now one dimensional. 

\begin{namedtheorem}[Asymptotic 1D]
\label{thm:asymp1D}
A heuristic solution to the optimization of parameters  $\tau_p$ and $p_a$ is
\begin{eqnarray}   
 {\tau_p^o}= \frac{\tau_u}{3}\quad \quad
 p_a^o K = b^o \sqrt{\tau_u M}
\end{eqnarray}
where  $b^o$ maximizes the sum rate $\cR_a$ where $\tau_p$ is set to $\tau_u/3$:
\begin{equation}
\cR_{a}^{[1D]}\!  =\!  b \,
\E_{\beta_0}\log_2 \! \left (\! 1+  
  \frac{ \sqrt{M \tau_u} {\beta}_0^2 /3 }{
       b  \overline{\beta^2}  M  +
               b^2 \overline{\beta}^2 \sqrt{M \tau_u}     +
        b  \overline{\beta}  {\beta}_0  \tau_u/3 }
\! \right)
\end{equation}
$b^o$ depends on $M$, $\tau_u$ and  the statistics of the $\beta$'s. 
\end{namedtheorem}

\begin{namedtheorem}[Heuristic Solution 2]
\label{thm:heuristic2}
A heuristic solution to the optimization of parameters  $\tau_p$ and $p_a$ is
\begin{eqnarray}   
 {\tau_p^o}= \frac{\tau_u}{3}\quad \quad
 p_a^o K = b^o \sqrt{\tau_u M}
\end{eqnarray}
where $b^o$ maximizes the heuristic cost function \eqref{eq:heur1} where  $\tau_p$ is set to $\tau_u/3$:
\begin{eqnarray}
\cR_{h}^{[1D]} = b \;
\E_{\beta_0}  \left[ \log_2 \left (1+  \frac{ \beta_0^2}{
                3 \overline{\beta} b^2} \right) \right]
\label{eq:heur2}
\end{eqnarray}
$b^o$ does not depend on $M$ or $\tau_u$ but solely on  the statistics of the $\beta$'s. 
\end{namedtheorem}

%

\section{Sum Rate Scaling Laws}
\label{sec:Scale}


In this section, we give a summary of the scaling laws of the optimized sum rate and optimized parameters based on $\cR_a$ and \eqref{eq:optG}. 
As seen through numerical results, $\cR_a$-based optimization gives results that are close to direct optimization based on $\cR_1$, 
provided that the variance of the $\beta$'s is not very large.
Hence such an analysis makes sense under such conditions. 
The details can be found in Appendix~\ref{sec:App2}. 

\begin{namedtheorem}[Scaling Laws]
\label{thm:ScalingLaws}
Assuming $\tau_u \gg 1$,  $M \gg 1$, 
the scaling laws of the optimized sum rate $\cR_a$ and optimized parameters  $\tau_p^o$ and  $p_a^o K$ are summarized as follows. 
Parameter $\tau_p$ is not constrained by the channel coherence time. 

\noindent \textbf{Case 1:} $M \gg \tau_u$:
\begin{eqnarray}
\hspace{-5mm}
\tau_p^o \approx \frac{\tau_u}{2} + O \! \left(\! \tau_u \sqrt{\frac{\tau_u}{M}} \, \right) \\
 p_a^o K \approx \sqrt{\frac{ \overline{\beta^4}}{\overline{\beta}^2 \overline{\beta^2}}}  \frac{1}{2} \sqrt{M \tau_u}  + O(\tau_u) \\
 \cR_a^o \approx \frac{\tau_u}{4 \log(2)} + O \! \left( \!  \tau_u \sqrt{\frac{\tau_u}{M}} \right)\\
  \underline{\SINR}_a \approx \sqrt{\frac{\tau_u}{M}}+ O \! \left(  {\frac{\tau_u}{M}} \right)
\end{eqnarray}

\noindent \textbf{Case 2:} $M \ll \tau_u$: 
\begin{eqnarray}
\hspace{-5mm}
\tau_p^o  = \left( \frac{M}{2}\right)^{2/3} \!\! \tau_u^{1/3} + o(\tau_u), \\
p_a^o K =  \sqrt{  \frac{ \overline{\beta^4}}{\overline{\beta}^2 \overline{\beta^2}}} 
\left({\frac{M}{2 \tau_u}} \right)^{5/6}\!\!  + O(M) \\
 \cR_a^o = 
 M 
 + 
 O \left(
  M \left[   \frac{ M}{\tau_u} \right]^{2/3} 
 \right)
\\
 \underline{\SINR}_a = 
 2^{1/3} \left({\frac{M}{\tau_u}} \right)^{1/6}
 +
O\left(   \left[ {\frac{M}{\tau_u}} \right]^{1/3}  \right)
\end{eqnarray}

\vspace{2mm}
\noindent \textbf{Case 3:} $M \sim \tau_u$: 
 \begin{eqnarray}
 \tau_p^o = a \tau_u,  \quad p_a^o K = b \sqrt{M \tau_u},  
\quad \quad  a=O(1), b=O(1) \\
 \cR_a^o = O(\sqrt{M \tau_u}),   \quad \underline{\SINR}_a = O(1)
\end{eqnarray}

\end{namedtheorem}

In Case 1, the average number of active devices  and the sum rate is limited by $\tau_u$. 
In Case 2, it is  is limited by $M$. 
In Case 3, the optimal number of pilot sequences and average number of active devices  become comparable. 
In Case 1 and Case 2, the rate of each device becomes asymptotically small but the average number of active devices that the system can accommodate grows faster. 
In Case 3, the rate of each device becomes constant while the average number of active devices increases. The quality of service requirement should dictate the choice of $M$ and $\tau_u$. 

{
In Case 1 and Case 2, we notice that the optimization results depend on the distribution of the parameters $\beta$ through the factor 
$\frac{ \overline{\beta^4}}{\overline{\beta}^2 \overline{\beta^2}}$. 
In Model 1, this factor has a small value range even for a wide range of values of $\alpha$. Hence, for this model, the optimization will not significantly depend on the distribution of the $\beta$'s. 
On the other hand, for Model 2 and Model 3, the factor $\frac{ \overline{\beta^4}}{\overline{\beta}^2 \overline{\beta^2}}$ varies significantly with varying values of $\sigma_v^2$ and $\alpha$ and so do the optimized values. For Case 3, we notice the same kind of behavior via simulations. Note that all those points are detailed in the simulation section. 
}

\section{Numerical Results}
\label{sec:numerical}

In this section we illustrate the  performance of the random pilot and data access protocol at optimized value of $\tau_p$ and $p_a K$. Due to the computational limitations in evaluating $\cR_1$, the value of $K$ is not as large as could be expected in an mMTC scenario. However, the value of $K$ is chosen so that  $\cR_1$ is close to its limit when $K$ grows to infinity (see Section~\ref{sec:Bound1}) so that our results are still representative. 
Furthermore, $\tau_u$ and 
$M$ are chosen to be of the same order. 

The curves show three different metrics as a function of the transmission slot duration  $\tau_u$ and  for  a fixed number of antennas $M$. 
The metrics are evaluated using the different optimization methods described in the previous sections. 
The metrics are: (1) Rate $\cR_1$ in bits per symbols; (2) The optimal number of pilot sequences; (3) The optimal activation probability through the quantity $p_a^o K$, i.e., the average number of active devices. 
For easy reference, the different optimization methods are summarized here: 
\begin{enumerate}
\item $\cR_1$-opt: direct optimization based on the main bound $\cR_1$.  
\item $\cR_3$-opt: optimization based on  bound $\cR_3$.  
\item $\cR_a$-opt: optimization based on  asymptotic bound $\cR_a$.  
\item $\cR_a^{[1D]}$-opt: $\tau_p^o$ is fixed to $\tau_u/3$,  $p_a^o K$  determined from  $\cR_a$. 
\item $\cR_h^{[0]}$-opt: $\tau_p^o=\tau_u/3$, $p_a^o K = \sqrt{\tau_u M} / \sqrt{ 3  s_0}$ in (\ref{eq:heur1}). 
\item $\cR_h^{[1D]}$-opt: $\tau_p^o=\tau_u/3$, $p_a^o K$ is based on the heuristic rate function  $\cR_h^{[1D]}$ in (\ref{eq:heur2}). 
\end{enumerate}
In all the figures, except Figure~\ref{fig:dB025_M400}, the number of antennas is $M=100$. 
We show performance for perfect power control and different values of the  variance of the $\beta$'s. 
For Model~1, we show only one curve as performance is quite insensitive to the variance of the $\beta$'s. 
Model~2 and Model~3 exhibit the same kind of behavior and we only show performance for one of the models to illustrate one particular performance aspect. 

\subsection{Optimization of Main Bound $\cR_1$}

We first show performance results with direct optimization of $\cR_1$. 
{Our main purpose is to show the variations of the optimized points w.r.t. the distribution of the $\beta$'s and w.r.t. to total number of devices $K$. Those points are highlighted in Section~\ref{sec:PerfBounds} and \ref{sec:Scale}. }
The heavy computations involved in the numerical evaluations have imposed a limit on the number of Monte-Carlo runs on parameters $\{\beta_j\}$  in \eqref{eq:cR1}. Hence, some curves fluctuate but still indicate significant performance results. 

In Figure~\ref{fig:R1opt_1}, for Model 1, the optimal sum rate is shown for the case of perfect power control ($\alpha=0$) and $\alpha=0.25$,  $\alpha=0.5$
for different values of the total number of devices $K$. We see that  performance is insensitive to the values of parameters $\alpha$ and  $K$. The optimized values of $\tau_p$ and $p_a K$ are also approximately the same. 
\begin{figure}
\centering
\includegraphics[width=9.1cm]{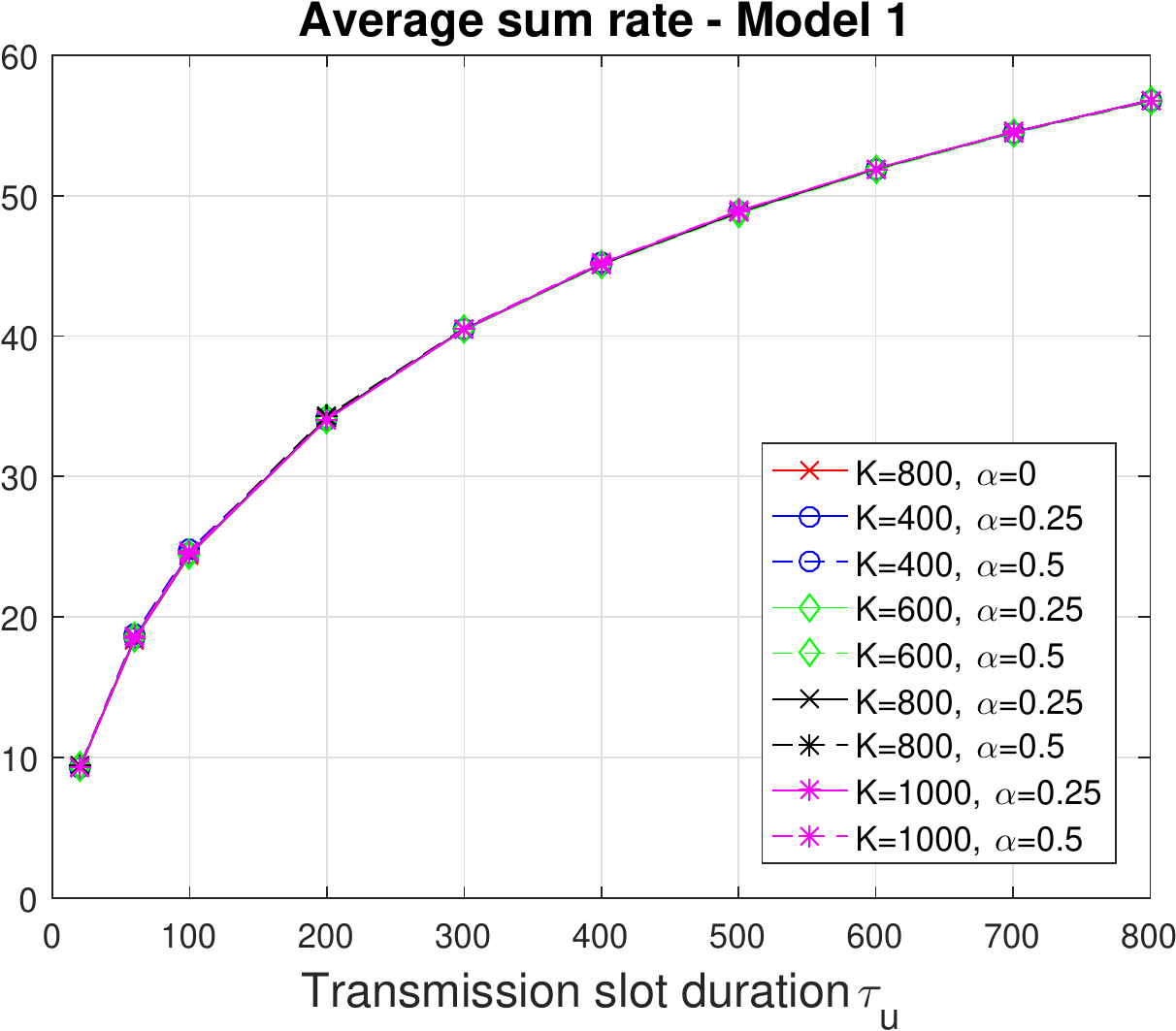}
\caption{Model 1. Optimal  sum rate $\cR_1$, for $M=100$, $K = \{400, 600, 800, 1000\}$, $\alpha=\{0, 0.25, 0.5\}$.}
\label{fig:R1opt_1}
\end{figure}

In Figure~\ref{fig:R1opt_2}, for Model 2, the optimal sum rate is shown for $\sigma_v^2=0$ and $\sigma_v^2=0.5$,  i.e., a large variation of the variance,  for different values of the total number of devices $K$. 
As mentioned in Section~\ref{sec:Bound1}, there is $K$ above which the rate does not change significantly.  
In numerical evaluations that are not shown in this paper, we also observed that the larger the variance of the $\beta$'s, the larger the number of devices to reach this  steady state. 
\begin{figure}
\centering
\includegraphics[width=8cm]{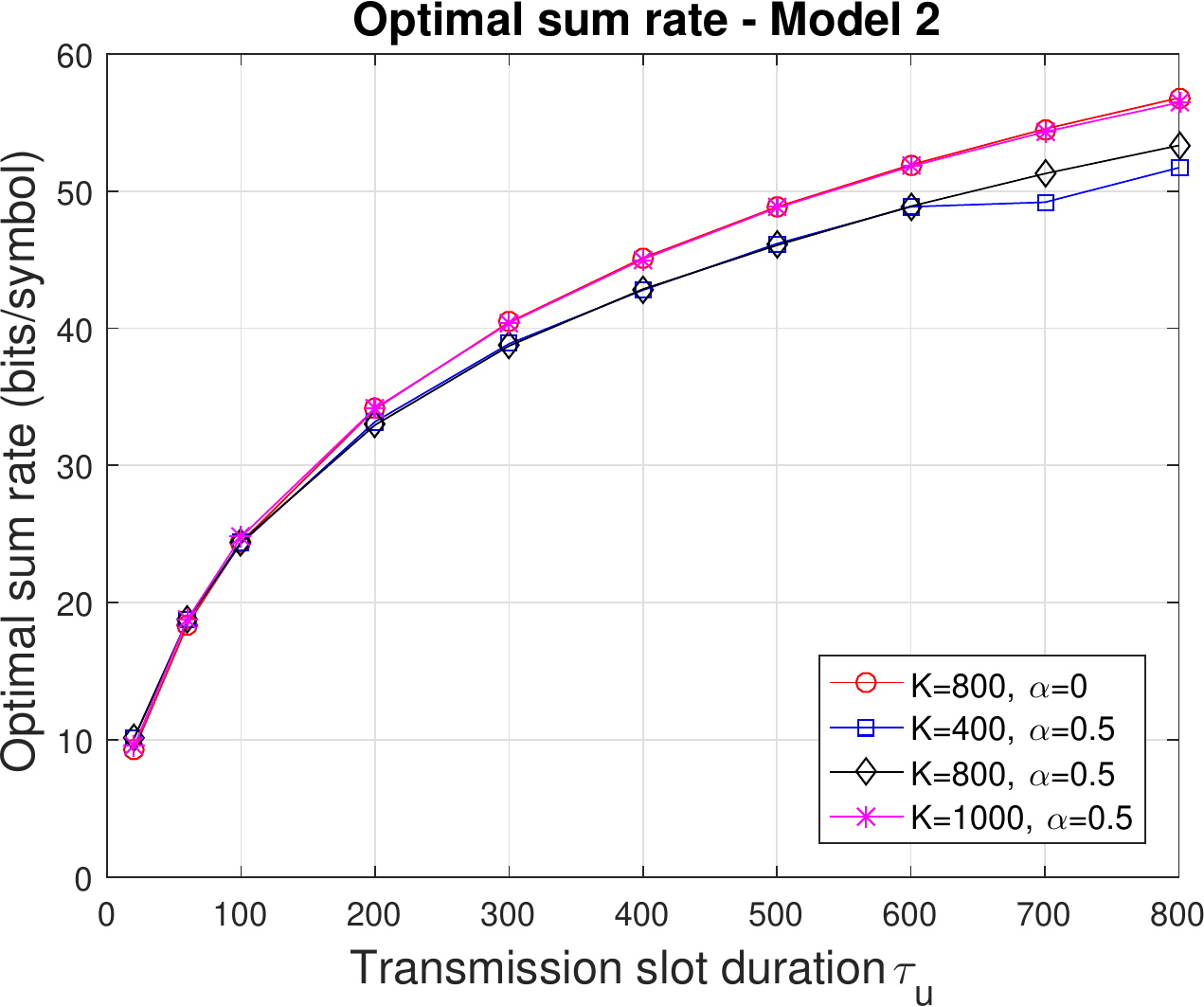}
\caption{Model 2. Optimal  sum rate $\cR_1$, for $M=100$, $K = \{400, 800, 1000\}$, $\sigma_n^2=\{0, 0.5\}$.}
\label{fig:R1opt_2}
\end{figure}

In Figure~\ref{fig:R1opt_3}, for Model 2, the sum rate curves, optimized $\tau_p$ and $p_a K$ are shown for different values of $\sigma_v^2$ and a large $K=1000$. 
Performance degrades as $\sigma_v^2$ increases. However, the degradation is small, even with large $\sigma_v^2$. 
The optimal number of pilot sequences appears quite insensitive to the variance of $\beta$, while the optimal average number of devices increases with the variance.
The same was observed for Model 3. 

\begin{figure}[h!]
        \centering
        \begin{subfigure}[b]{\columnwidth} \centering 
\includegraphics[width=8cm]{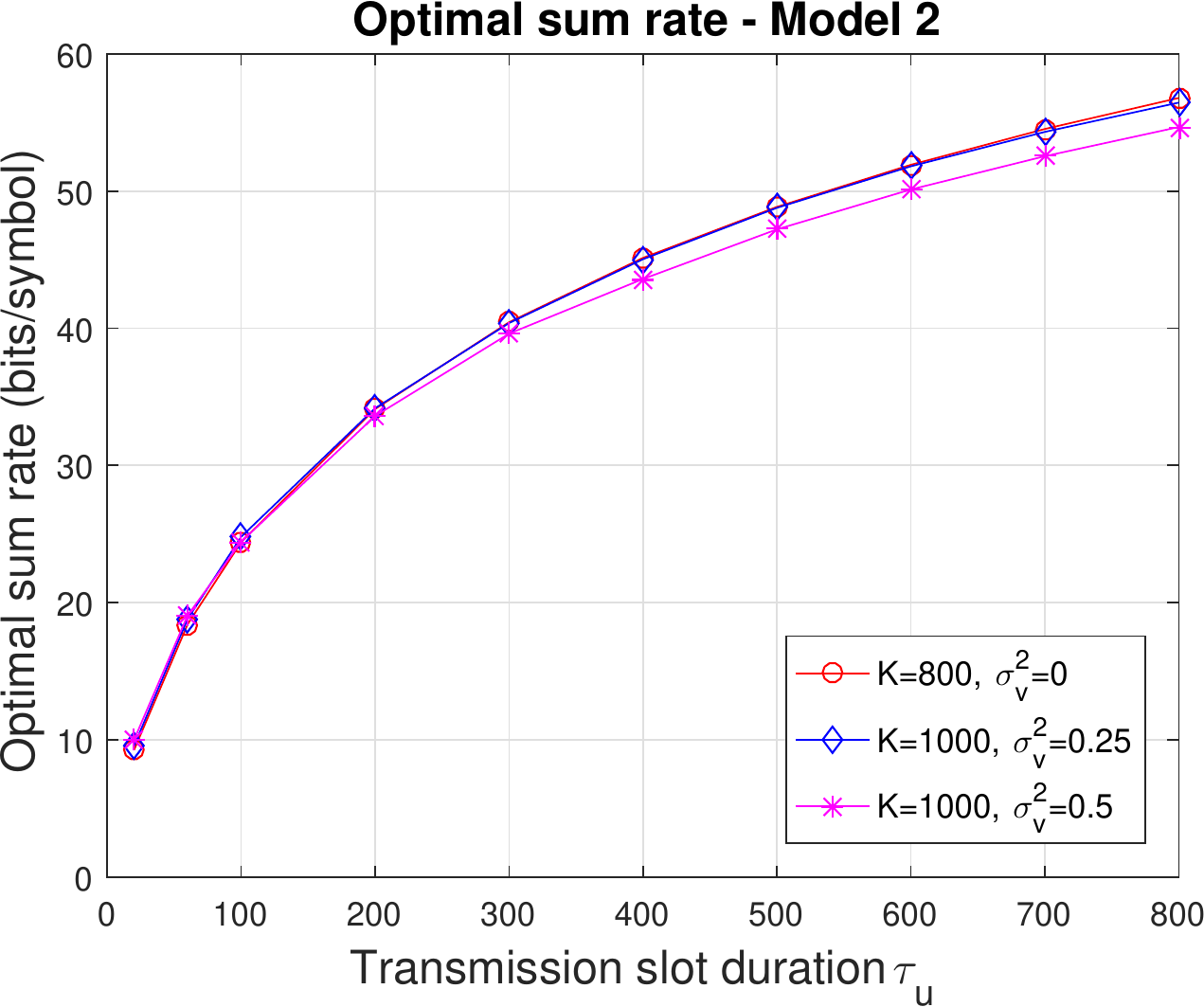}
                \caption{}
        \end{subfigure} \\ [5mm]
        \begin{subfigure}[b]{\columnwidth} \centering
\includegraphics[width=8cm]{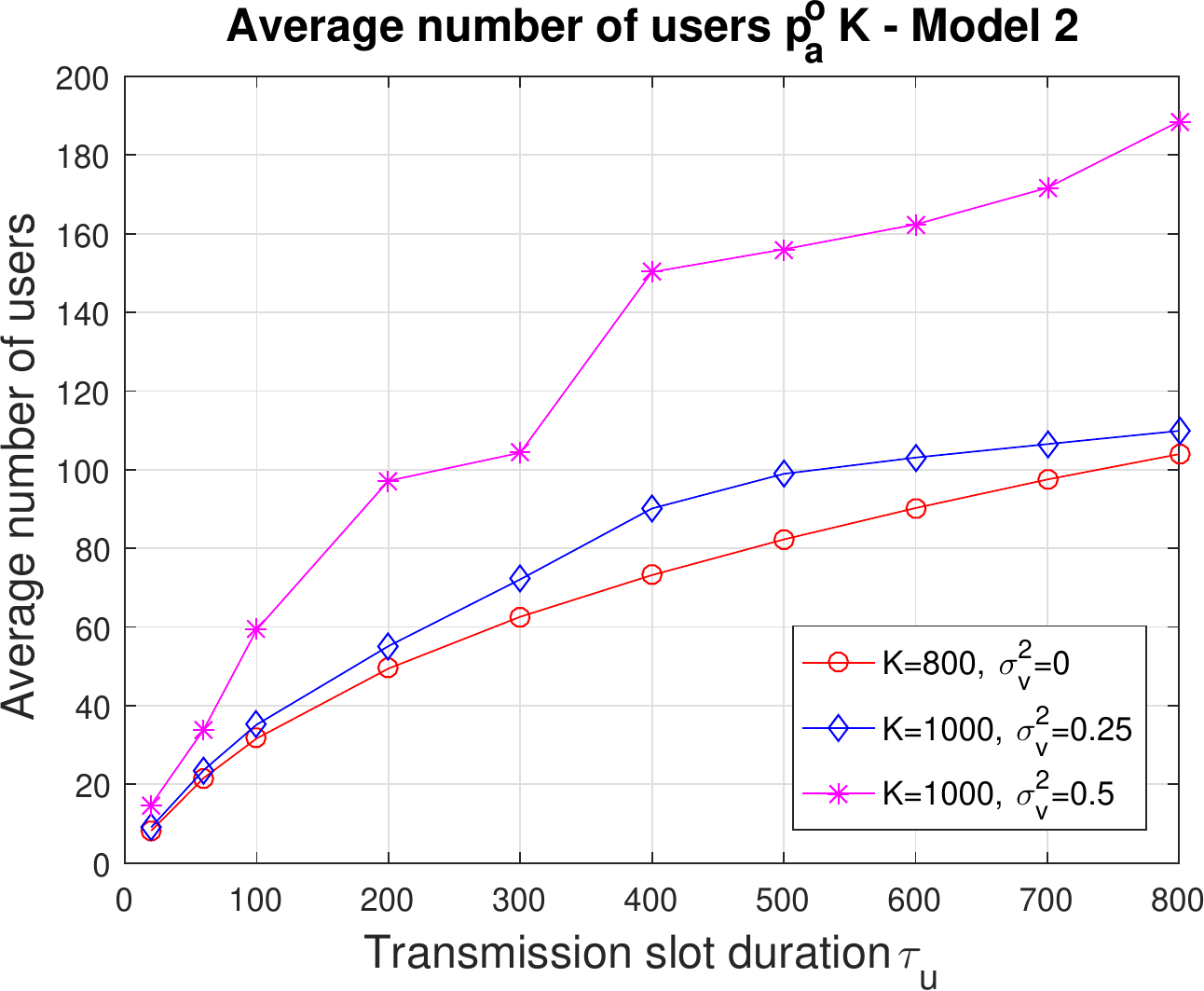}
                \caption{}
        \end{subfigure} \\  [5mm]
             \begin{subfigure}[b]{\columnwidth} \centering
\includegraphics[width=8cm]{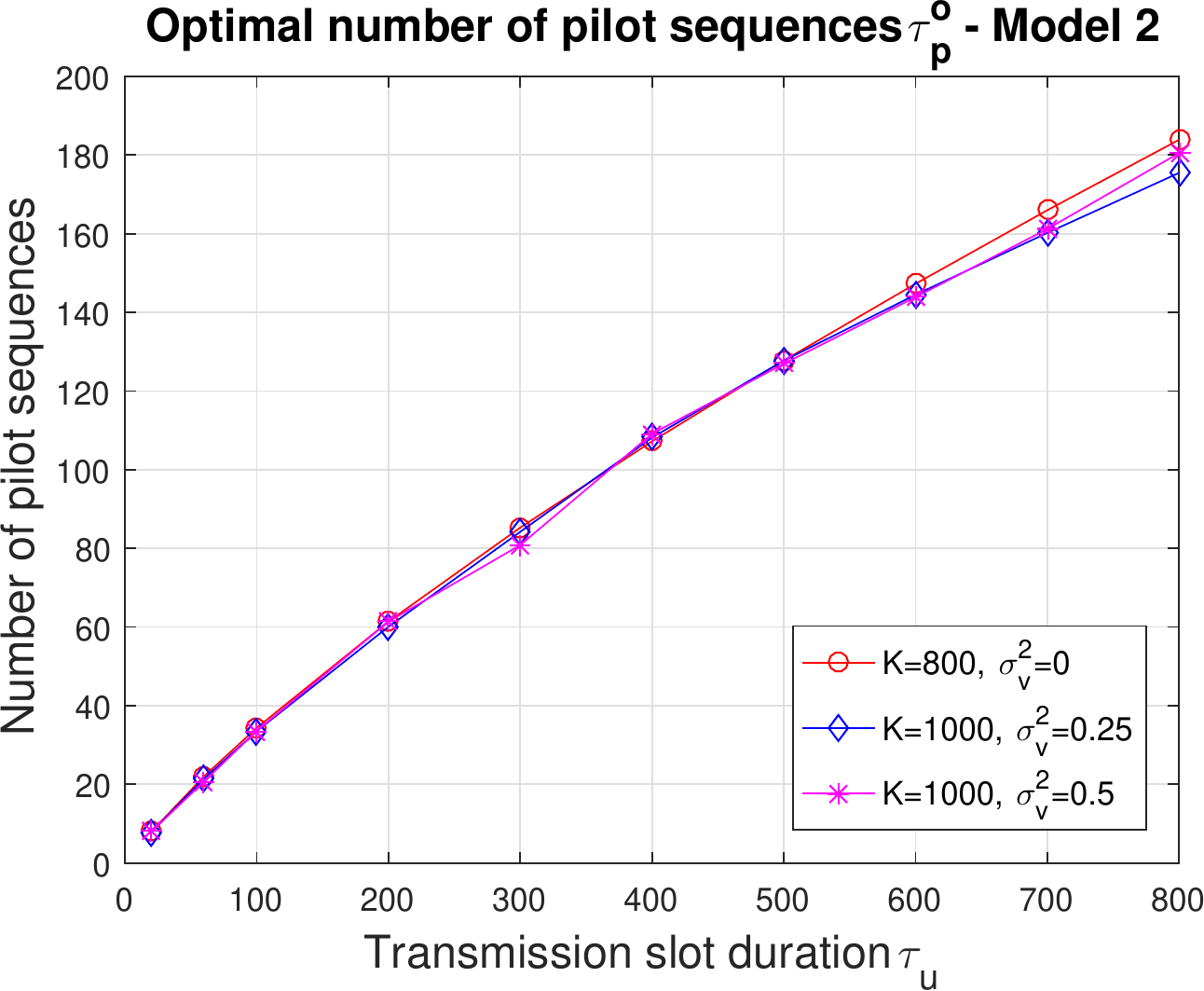}
                \caption{}
        \end{subfigure} 
        \caption{
Model 2. (a) Optimal  sum rate $\cR_1$,  (b) $p_a K$ and (c) $\tau_p$, for $M=100$, $K = \{800, 1000\}$, $\sigma_n^2=\{0, 0.25, 0.5\}$.         }
\label{fig:R1opt_3}
\end{figure}

\subsection{Optimization Methods}

We now show the performance metrics using the six different methods previously described. 
In Figures~\ref{fig:Dist025}, the sum rate and optimized $\tau_p$ is shown for  Model 3, $\alpha=0.25$. 
In Figures~\ref{fig:Dist05}, for the same model,  $\alpha=0.5$,  the sum rate and optimized $p_a K$ is shown.
The total number of devices $K$ is equal to 800. 
For all the parameter settings that we tested, the performance using $\cR_3$ or $\cR_a$ are almost the same, so that we refer only to $\cR_a$ from now on. 
\begin{figure}[h]
        \centering
        \begin{subfigure}[b]{\columnwidth} \centering 
\includegraphics[width=8cm]{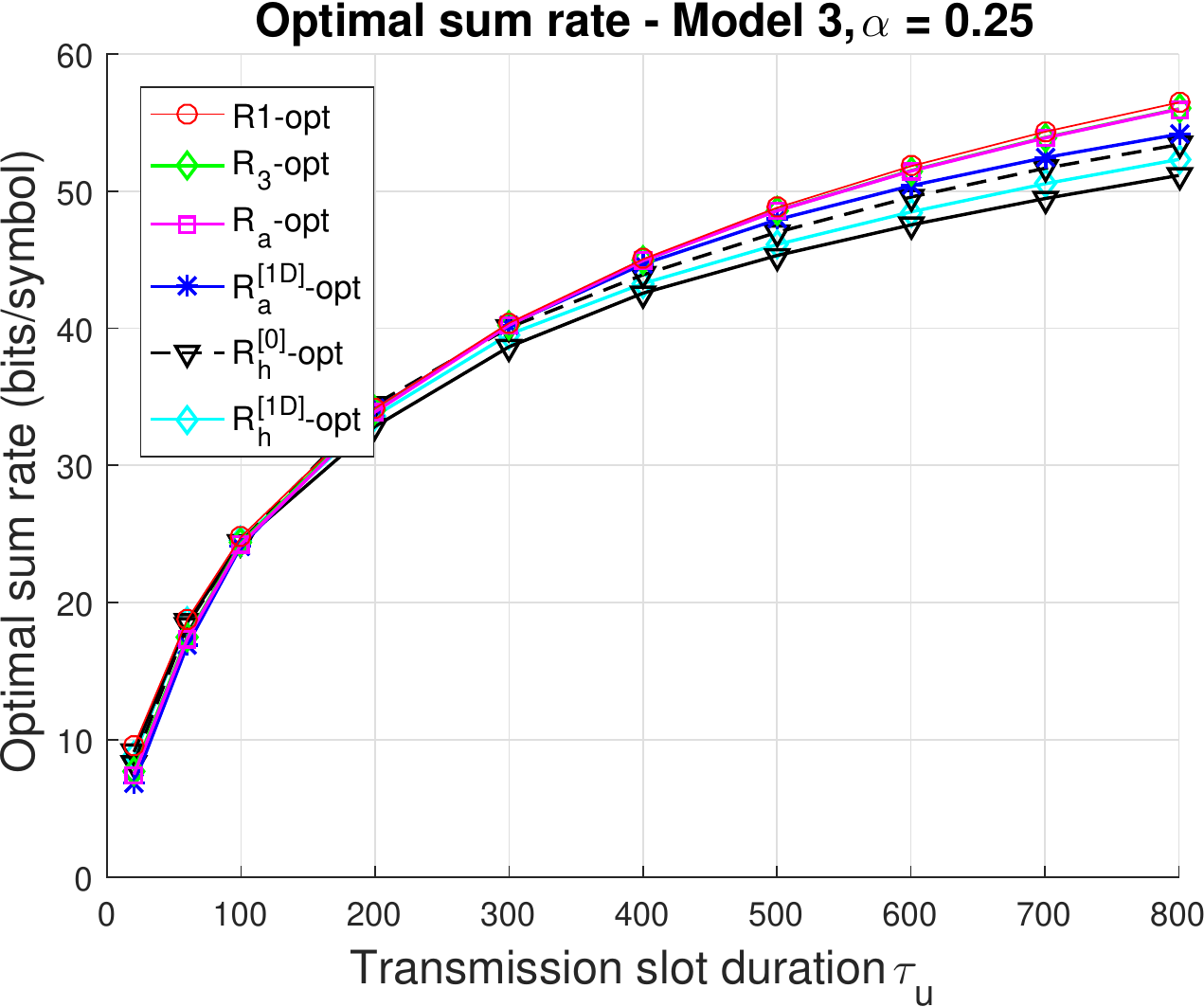}
                \caption{}
        \end{subfigure} \\ [5mm]
        \begin{subfigure}[b]{\columnwidth} \centering
\includegraphics[width=8cm]{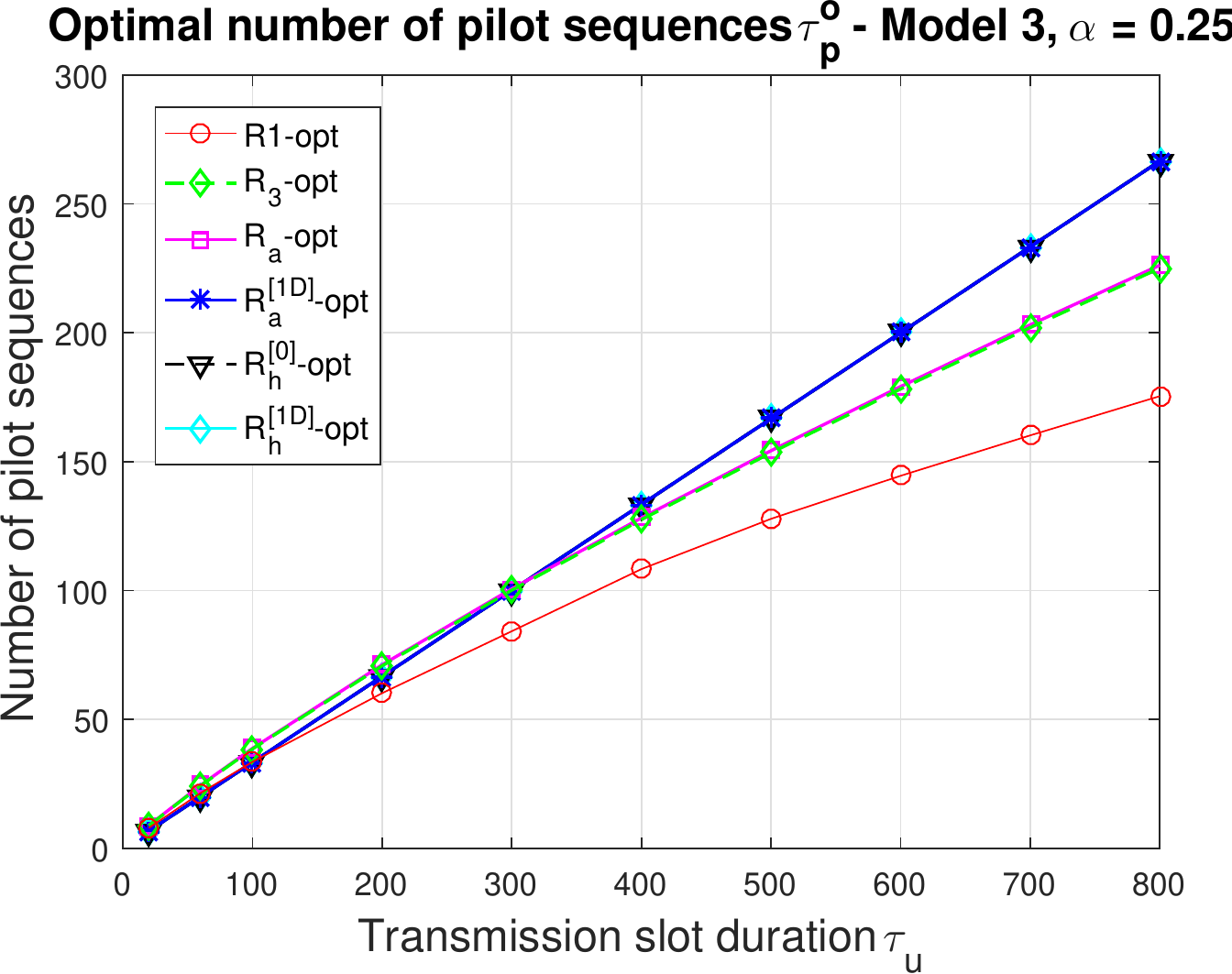}
                \caption{}
        \end{subfigure} \\ 
        \caption{
Model 3. (a) Optimal  sum rate $\cR_1$ and (b) $\tau_p$, for $M=100$, $K = 800$, $\alpha= 0.25$. }
\label{fig:Dist025}
\end{figure}

For the medium value $\alpha=0.25$, $\cR_a$-opt follows closely $\cR_1$-opt. There is a slight degradation for the heuristic solutions for larger values of $\tau_u$ while heuristic solutions perform slightly better for smaller values of $\tau_u$. 
This behavior is confirmed and magnified for $\alpha=0.5$. 
The optimal number of pilot sequences tends to be overestimated compared to $\cR_1$-opt when the  transmission slot duration becomes large. 
The gap  in the estimation of $p_a K$ is quite large, but this does not show significantly in the sum rate $\cR_1$. 
This can be explained by the fact that $\cR_1$ is quite flat around the optimal point. 

\begin{figure}[t]
        \centering
        \begin{subfigure}[b]{\columnwidth} \centering 
\includegraphics[width=8cm]{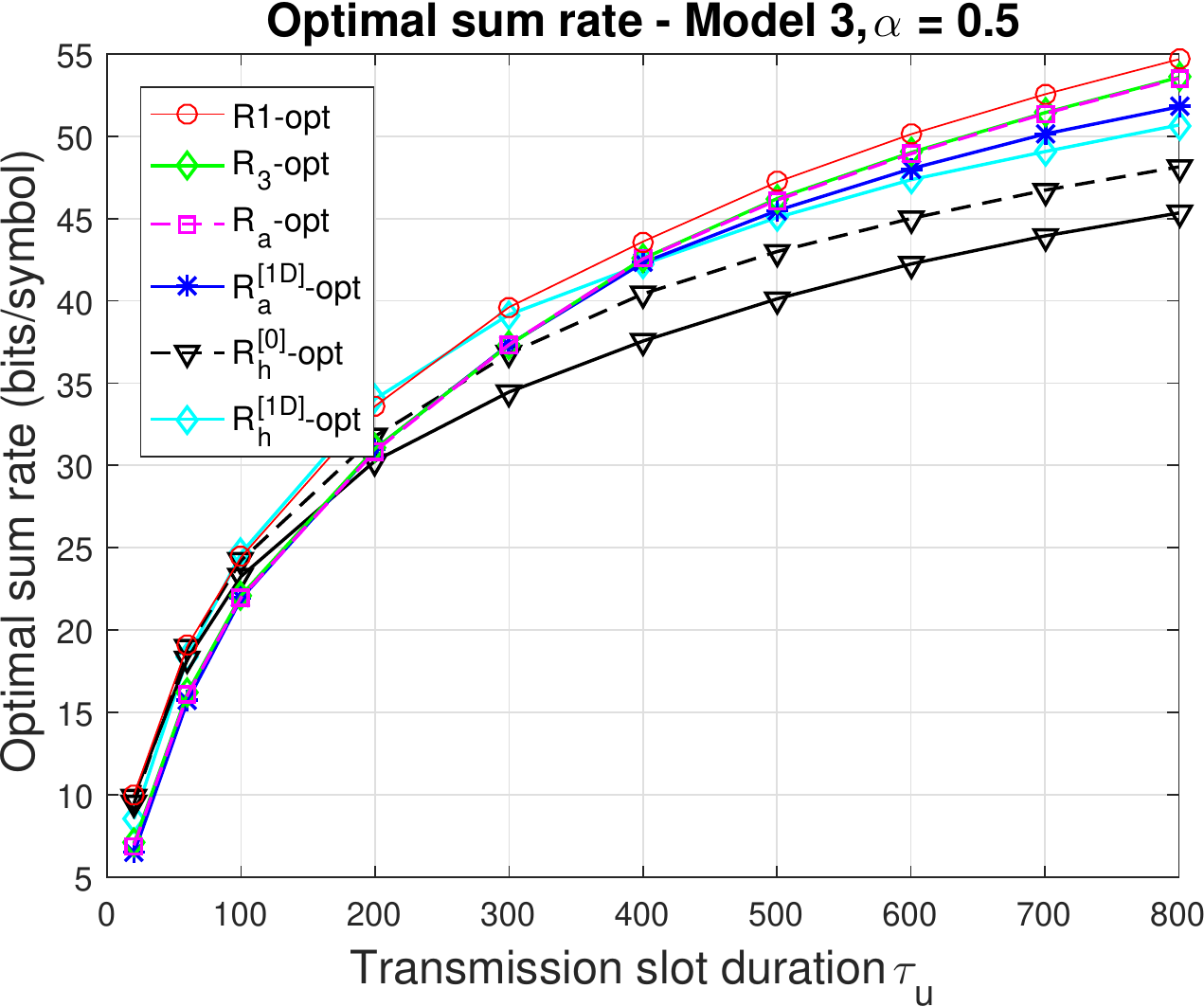}
                \caption{}
        \end{subfigure} \\ [5mm]
        \begin{subfigure}[b]{\columnwidth} \centering
\includegraphics[width=8cm]{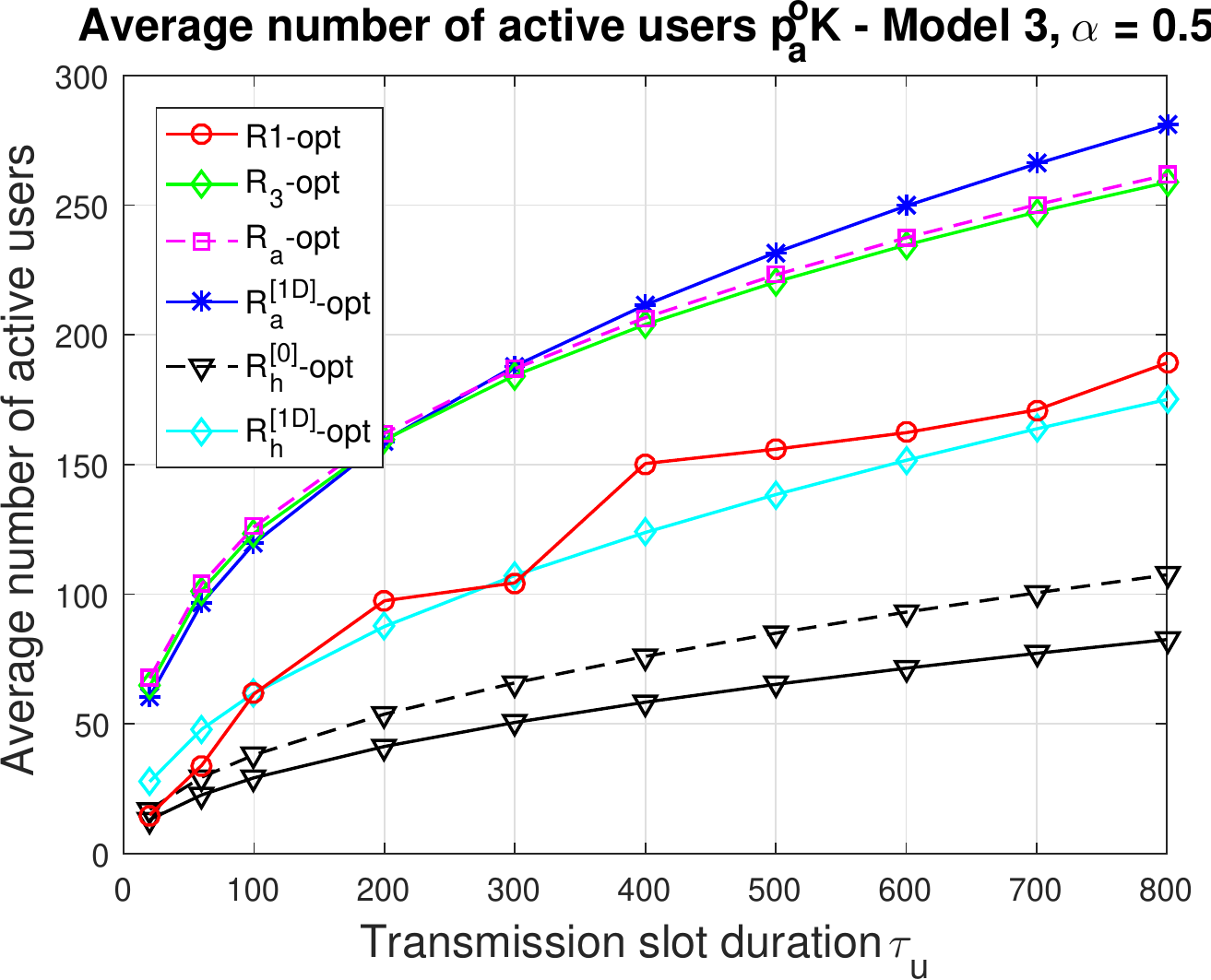}
                \caption{}
        \end{subfigure} \\ 
        \caption{
Model 3. (a) Optimal  sum rate $\cR_1$ and (b) $p_a K$, for $M=100$, $K = 800$, $\alpha= 0.5$. }
\label{fig:Dist05}
\end{figure}

\begin{figure}
\centering
\includegraphics[width=8cm]{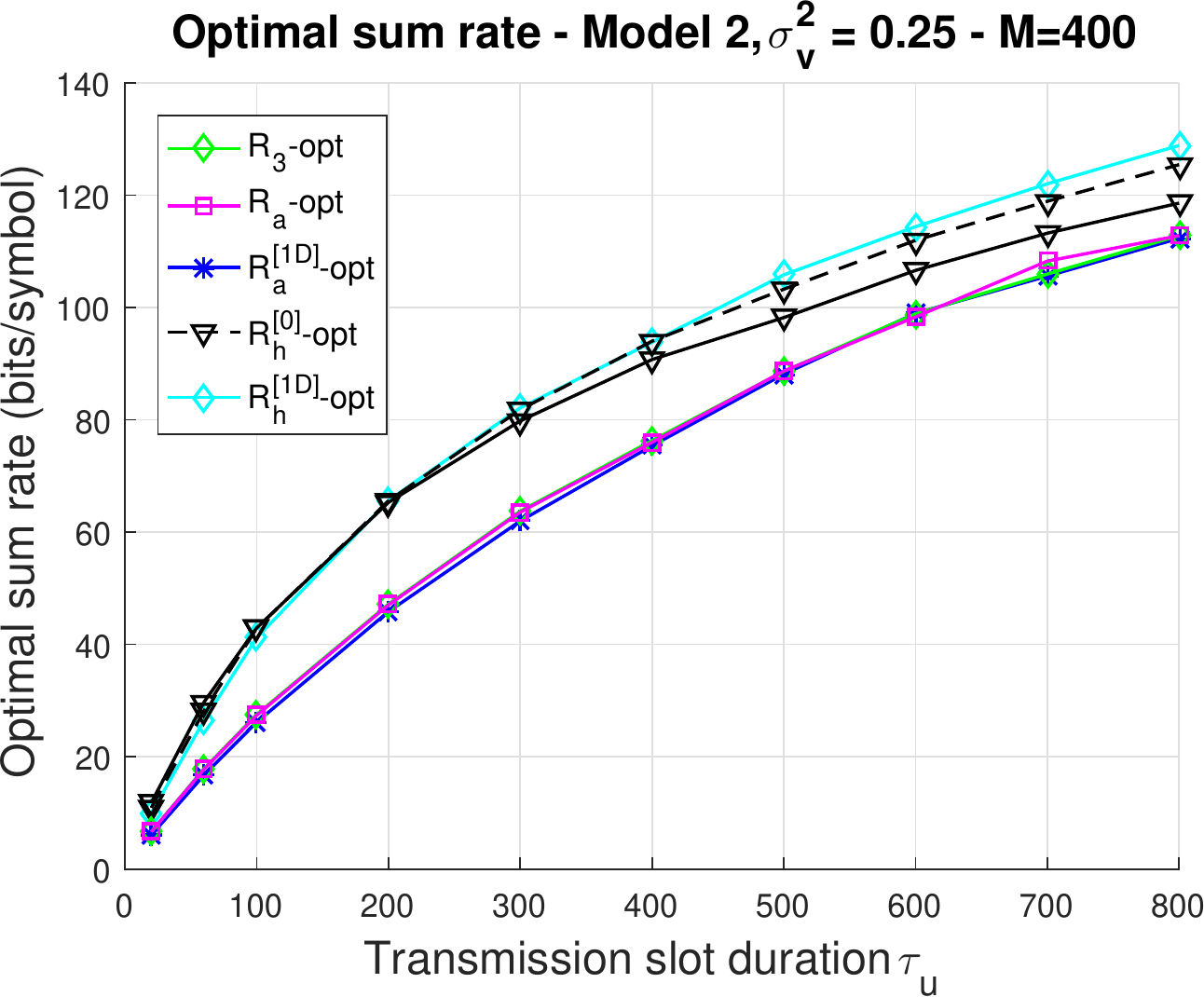}

\caption{Model 2. Optimal average sum rate $\cR_1$ for $M=400$, $K = 800$, $\sigma_v^2= 0.25$. \label{fig:dB025_M400}}
\end{figure}


At last,  we show a more extreme case, with $M=400$, Model 2 and $\sigma_v^2=0.25$, where in  the selected range of $\tau_u$, 
$\cR_a$ based methods perform significantly worse than the heuristic methods. 
In such a case, power control is the preferred solution as it leads to a robust and simple optimization. 
If power control is not possible, it is clear that resorting to direct optimization of $\cR_1$ is preferable, while heuristic solutions could be of significant value.

\section{Conclusion}
Massive MIMO is recognized as a primary technology in 5G for its ability to provide very high data rates. It is also a major enabling tool for machine-type communications where the sheer number of devices can be accommodated relying on the high multiplexing gain of massive MIMO and low-power devices can be served relying on its large array gain. In this paper, we have proposed a joint pilot and data transmission protocol based on random access that is adapted to the intermittent activity pattern of the devices. This protocol is organized in transmission slots 
and relies on the averaging of the pilot collision events across the transmission slots. 
It is suited for delay-tolerant and low-rate applications. 
We have provided performance expressions as well as optimization tools that are particularly important for a system where the activity of the devices and the number of pilots have to obey certain statistical rules.

\appendices

\section{Proof of Theorem~\ref{theorem:first-theorem}}
\label{sec:proof1}

The derivation of $\underline{\SINR}_1(\{\beta\})$ follows~\cite{Ngo2013a,Bjornson2016a}. The pre-log term  in~\eqref{eq:cR1}  accounts for the activity probability of devices.  The main derivation steps are reported below. 

\subsection{Derivation of $\underline{\SINR}_1$}

We first compute the achievable sum rate for a fixed number of active devices $K_a$.
The $M\times 1$ vectorial received signal in the data phase is
\begin{equation}
{\mathbf{y}_d}  =
 \sum_{i =1}^{K_a} \mathbf{g}_k x_k      + \mathbf{n}
\end{equation}
where $x_k$ is the transmitted symbol from device $k$ and $\mathbf{n}$ is the additive  noise vector.
We isolate the contributions with index set $ \{0, {\mathcal C}_0 \}$ and its complement and write 
${\mathbf{g}}_k$ as  $\hat{\mathbf{g}}_k -  \bm{\varepsilon}_k$, 
where $\bm{\varepsilon}_k$ is the  error vector in the MMSE estimation.
The received signal can be rewritten as
\begin{eqnarray}
{\mathbf{y}_d}    = 
%
   \sum_{j \in \{0, {\mathcal C}_0 \} }  ( \hat{\mathbf{g}}_j -  \bm{\varepsilon}_j) x_j  + 
  \sum_{j \notin \{0, {\mathcal C}_0 \} } \mathbf{g}_j x_j      + \mathbf{n}  
\end{eqnarray}
The BS applies MRC based on channel estimate $\hat{\mathbf{g}}_0$ in (\ref{eq:mmse}):
 \begin{eqnarray}
\hat{\mathbf{g}}_0^H {\mathbf{y}_d} =  
  \hat{\mathbf{g}}_0^H \hat{\mathbf{g}}_0 x_0
+
  \sum_{j \in  {\mathcal C}_0  }   \hat{\mathbf{g}}_0^H \hat{\mathbf{g}}_j x_j  
 + \\
 \sum_{j \in \{j, {\mathcal C}_0 \} }   \hat{\mathbf{g}}_0^H{\bm{\varepsilon}}_j x_j  + 
 \sum_{j \notin \{0, {\mathcal C}_0 \} }  \hat{\mathbf{g}}_0^H\mathbf{g}_j x_j 
 +
\hat{\mathbf{g}}_0^H \mathbf{n} 
\end{eqnarray}
where  $  \hat{\mathbf{g}}_0^H \hat{\mathbf{g}}_0 x_0$ is the signal of interest. The other additive terms are treated as a Gaussian noise. Using the independence of $\bm{\varepsilon}_j$ and 
 $\hat{\mathbf{g}}_0$ in MMSE estimation, those terms are independent from the signal of interest. 
From (\ref{eq:mmse3}), the achievable rate is
\begin{equation}
\underline{R} = 
\E_{\{\bfh_k\}}
\log_2 \left( 1+\underline{\SINR} \right) 
\label{eq:Capa}
\end{equation}
where $\underline{\SINR}$ is equal to
\begin{eqnarray}
\frac{{} |\hat{\mathbf{g}}_0^H  \hat{\mathbf{g}}_0 |^2}
{
{}  \displaystyle{\sum_{j \in  {\mathcal C}_0  }}    |  \hat{\mathbf{g}}_0^H \hat{\mathbf{g}}_i  |^2
 + 
 {} \hspace{-3mm} \sum_{j \in \{k, {\mathcal C}_0 \} }  \hspace{-3mm} |  \hat{\mathbf{g}}_0^H{\bm{\varepsilon}}_i  |^2 +
 {} \hspace{-3mm} \sum_{j \notin \{k, {\mathcal C}_0 \} } \hspace{-3mm} |\hat{\mathbf{g}}_0^H\mathbf{g}_i |^2+ 
  \|\hat{\mathbf{g}}_0 \|^2 
}.
\end{eqnarray}
The notation $\E_{\{\bfh_k\}}$ indicates the expected value w.r.t. small scale fading terms $\bfh_k$, defined as $\bfg_k = \sqrt{\beta_k} \bfh_k$.

Following~\cite{Ngo2013a}, using Jensen's inequality on function $\log_2(1+1/x)$ allows averaging over  the channel fast fading of the interferers. 
We obtain a lower bound as:  
\begin{equation}
\underline{R}_1(\{\beta\})= 
\log_2 \left( 1+\underline{\SINR}_1(\{\beta\})\right) 
\label{eq:capaP}
\end{equation}
where $1/\underline{\SINR}_1(\{\beta\})$ is equal to
\begin{eqnarray}
\frac{  \sum_{j \in \mathcal{C}_0}  \hspace{-1mm} \beta_j^2}
{ \beta_0^2 }
+ 
 \frac{  \sum_{j \in  \{0, \cC_0  \} } \hspace{-1mm} \sigma^2_{\bm{\varepsilon}_j}  }
 {(M-1) \sigma^2_{\hat{\mathbf{g}}_0} } 
 + 
 \frac{    \sum_{j \notin \{0, \cC_0  \}}  \hspace{-1mm} \beta_j  + 1 }
 {(M-1) \sigma^2_{\hat{\mathbf{g}}_0} }.
\label{eq:snrP}
\end{eqnarray}
Replacing the expression of $\sigma^2_{\bm{\varepsilon}_j}$ and  $\sigma^2_{\hat{\mathbf{g}}_0}$  in  (\ref{eq:mmse3}) and (\ref{eq:mmse2}), 
we obtain (\ref{eq:sinr1}). 

The first term in \eqref{eq:snrP} represents the interference induced by pilot contamination, the 
second  comes from channel estimation errors and the third term comes is due to noise and residual interference after MRC.

\subsection{Derivation of the pre-log term}

(\ref{eq:capaP})  gives a lower bound on the maximal achievable rate of one given device $0$ with a population of $K_a$ active devices, with exactly $c$ contaminators and a given realization of the $\beta$'s.
This bound is tight compared to the maximal achievable rate thanks to channel hardening. 
Next, we give the different steps leading to \eqref{eq:cR1}.
Note that when the notations become too heavy and when there is no ambiguity, we give up the dependency w.r.t. the parameters. 
Furthermore, the notations in equations \eqref{eq:cR1} and \eqref{eq:R1} have been simplified compared to the notations adopted in this appendix. 

\begin{itemize}

\item We compute the achievable rate for device $0$ averaged over the possible sets of contaminators  within the set of $K_a-1$ active devices excluding device $0$. This last set is denoted as  ${\cal K}_a^{0}(\{\beta\})$ and the set of contaminators is denoted as 
 $\cC_0(c, {\cal K}_a^{0}(\{\beta\}))$.
 The rate is
\begin{equation}
\E_{\cC_0} \!  \left[ \underline{R}_1 \! \left( \beta_0,\cC_0(c, {\cal K}_a^{0}(\{\beta\})) | {\cal K}_a^{0}(\{\beta\}) \right)
\right]
\label{eq:step1}
\end{equation}
where $\E_{\cC_0}$ represents the expected value w.r.t. all possible sets of $c$ contaminators within  ${\cal K}_a^{0}$. 

\item  Next, we compute the expected value of \eqref{eq:step1} accounting for the random selection of the pilot sequences within the set ${\cal K}_a^0$ and obtain
\begin{equation}
\underline{R}_1^c \! \!  =  \!\! \sum_{c=0}^{K_a-1} p(c|{\cal K}_a^0) \E_{\cC_0} \!  \left[  \underline{R}_1 \! \left(  \beta_0, \cC_0(c, {\cal K}_a^{0}(\{\beta\}))| {\cal K}_a^{0}(\{\beta\}) \right) \right]
\label{eq:step2}
\end{equation}
where
$p(c|{\cal K}_a^0)$ is the probability of having $c$ contaminators to device $0$ within the set ${\cal K}_a^0$.

\item  Next, we compute the expected value of \eqref{eq:step2} accounting for the activation probability, i.e., the probability to have $K_a-1$ active devices within a total number of $K-1$ devices (as device $0$ is excluded). We get
\begin{eqnarray}
\sum_{K_a=1}^{K}
p(K_a-1|K-1) \E_{ {\cal K}_a^0 } \!  \left[ \underline{R}_1^c   \right]
\label{eq:step3}
\end{eqnarray}
where $\E_{ {\cal K}_a^0}$ represents the expected value w.r.t. all possible sets of $K_a-1$ active devices excluding device~$0$. 

\item  Now, accounting for the activation probability $p_a$ of device $0$ and noting that $p_a p(K_a-1|K-1) =  \frac{K_a}{K} p(K_a|K) $, the achievable rate for device $0$ averaged over the contamination events and activation probability is
\begin{eqnarray}
\underline{R}_1^0 (\beta_0)  \!\! = \!\!\frac{1}{K}
\sum_{K_a=1}^{K}
p(K_a|K) K_a
\!\!   \sum_{c=0}^{Ka-1} p(c|{\cal K}_a^0) 
\E_{ {\cal K}_a^0, \cC_0}  
\!  \left[ \underline{R}_1 \right]
\label{eq:step4}
\end{eqnarray}

\item   Finally, the sum rate is 
$
 \sum_{k=1}^K
\underline{R}_1^0 (\beta_k)
$ equal to 
\begin{eqnarray}
\sum_{K_a=1}^{K}
p(K_a|K) K_a
\sum_{c=0}^{Ka-1} p(c|{\cal K}_a^0) \\
\frac{1}{K} \sum_{k=1}^K
\E_{ {\cal K}_a^k, \cC_k}  \!  \left[ 
\underline{R}_1 \left(  \beta_k, \cC_k(c, {\cal K}_a^{k}(\{\beta\}))| {\cal K}_a^{k}(\{\beta\}) \right) \right]
\label{eq:step5}
\end{eqnarray}

\item   We make the approximation
\begin{eqnarray}
\frac{1}{K} \sum_{k=1}^K \E_{ {\cal K}_a^k , \cC_k}
 \!  \left[  \underline{R}_1 \right] \approx
\E_{\{\beta\}}   \!  \left[  
\underline{R}_1 \right]
\label{eq:step6}
\end{eqnarray}
where $\E_{\{\beta\}} $ is the expected value  of all $\beta_k$s with $k \neq 0$ w.r.t. 
their probability density function. 
Given that the $\beta_k$'s are all independent,
this approximation becomes better as $K$ grows large.
Expression \eqref{eq:step6} gives us \eqref{eq:cR1}.
\end{itemize}

\section{Scaling Laws}
\label{sec:App2}

An asymptotic analysis is performed with the three cases: 
\begin{itemize}
\item Case 1: $M \gg  \tau_u$, i.e.,  ${\tau_u}/{M}$ close to 0. 
\item Case 2: $M \ll  \tau_u$, i.e.,  ${M}/{\tau_u}$ close to 0. 
\item Case 3: $M \sim  \tau_u$, i.e., $M$ and $\tau_u$ are of the same order. 
\end{itemize}
The analysis is directly based on the asymptotic expression of the sum rate  $\cR_a$ in \eqref{eq:cRa}, instead of $\cR_3$,
which is justified by examining the asymptotic conditions $\tau_u \gg 1$ and $M \gg 1$ and the ones considered below.
Most of the results are based on first or second order Taylor approximations. To ease the presentation, we write the order of the approximations only for  the final optimized quantities.  
Furthermore, $\tau_p$ is not constrained by the channel coherence time in Appendix~\ref{sec:Case1}, Appendix~\ref{sec:Case2}, Appendix~\ref{sec:Case3}.
We treat this case in Appendix~\ref{sec:Case4}.

\subsection{Case 1: $M \gg  \tau_u$ (or $M \gg  \tau_p$)}
\label{sec:Case1}

We consider different variation domains for $p_a K$  and $\tau_p$ in order to study the optimization of the rate function $\cR_a$.

\subsubsection{$p_a K \gg \tau_p$}

\label{sec:Case11}
In the expression $\cR_a$, 
$\underline{\SINR}_a$ can be approximated as:
\begin{eqnarray}
\underline{\SINR}_a \approx \frac{M \tau_p \beta_0^2}
{ \overline{\beta^2} M K p_a + \overline{\beta}^2 p_a^2 K^2}  
\end{eqnarray}
provided that  $\overline{\beta^2}M  \gg \overline{\beta} \beta_0 \tau_p$. 
Furthermore, if $\tau_p {\beta_0^2}   \ll  \overline{\beta}^2 p_a K $, then 
$\underline{\SINR}_a \ll 1$. 

The two conditions on $\beta_0$
are verified for Model 1 and Model 3 (see Section~\ref{sec:Models}) provided $\alpha$  is such that $\beta_0$ keeps  values of same order as 
$\overline{\beta^2}/\overline{\beta}$ and $\overline{\beta}$.
For Model 2, we assume that  the largest value of $\beta_0$, denoted as $\beta_T$, verifying the two conditions is sufficiently large so that 
the expected value of the log-term in (\ref{eq:cRa}) restricted to domain $[0, \beta_T]$
is approximately  equal to the expected value over the whole domain.  
This result is used below where the expected value w.r.t. $\beta_0$ over the restricted domain is approximated as $\overline{\beta}$.
The same type of conditions on $\beta_0$ are met in some of the derivations that will follow, and can be treated in the same way so that 
we directly consider the case of perfect power control to simplify.

Next, we proceed to Taylor series approximations in order to find the optimal values of parameter $p_a K$ first and then $\tau_p$. 
We assume $\overline{\beta^2} M \gg  \overline{\beta}^2 p_a K$, equivalent to  $M \gg  p_a K$ for the models that we consider. A second order Taylor approximation of the term in $\log$ gives: 
\begin{equation}
\cR_a \log(2) \approx
\frac{\tau_u -\tau_p}{\tau_u} \E_{\beta_0}\! \left[
\frac{M \tau_p \beta_0^2 }
{ \overline{\beta^2} M  + \overline{\beta}^2 p_a K}\!  - \!
\frac{1}{2}\! \left( \frac{ \tau_p \beta_0^2 }
{ \overline{\beta^2}  } \right)^2 \!\! \! \frac{1}{K p_a}
\right]
\label{eq:Ra_approx}
\end{equation}
Taking the derivative w.r.t. the variable $p_a K$, we get:
\begin{equation}
\frac{\partial \cR_a}{\partial p_a K}  \approx 0 \Leftrightarrow
 \E_{\beta_0}\! \left[ - \frac{M \tau_p \beta_0^2 \overline{\beta}^2 }
{ [ \overline{\beta^2} M + \overline{\beta}^2 p_a K ]^2 } +
\frac{1}{2}  \frac{ \tau_p^2 \beta_0^4 } 
{ {\overline{\beta^2}}^2  p_a^2 K^2} \right] \approx 0
\end{equation}
We assume first $\overline{\beta^2} M \gg  \overline{\beta}^2 p_a K$. The other cases are treated later. 
We obtain the expression of the optimal value of $p_a K$:
\begin{eqnarray}
p_a^o K  = \sqrt{\frac{1}{2}   \frac{ \overline{\beta^4}}{\overline{\beta}^2 \overline{\beta^2}}} \sqrt{M \tau_p} + 
O(\tau_p ).
\label{eq:paK_opt}
\end{eqnarray}
Next, keeping only the first order terms in $\tau_p$  in \eqref{eq:Ra_approx}, the optimal value of  $\tau_p$ is the one maximizing $(\tau_u-\tau_p) \tau_p$, i.e., 
\begin{eqnarray}
\tau_p = \frac{\tau_u}{2} + O\left(\tau_u \sqrt{\frac{\tau_u}{M}}\right). 
\label{eq:tp_opt}
\end{eqnarray}

It can easily be verified that solutions $p_a^o K$ and $\tau_p^o$  are within the definition intervals $p_a K \gg \tau_p$, $\overline{\beta^2} M \gg  \overline{\beta}^2 p_a K$ for $M \gg \tau_u$.
The maximum value of the sum rate is 
\begin{eqnarray}
\cR_a^o = \frac{\tau_u}{4  \log(2)} + O\left( \tau_u\sqrt{\frac{\tau_u}{M}}\right). 
\end{eqnarray}

The last step is to verify that $\cR_a$ does not get larger values for the subcase where 
$p_a K \gg M$ or $p_a K \sim M$. For $p_a K \gg M$:
\begin{eqnarray}
\cR_a \approx
\frac{\tau_u - \tau_p}{\tau_u } p_a K \log_2\left( 1+ \frac{M \tau_p}{p_a^2 K^2}  \right).
\end{eqnarray}
For a fixed $\tau_p$, $\cR_a$ has its maximal value for $p_a^o K  =  \sqrt{M \tau_p}/\sqrt{3 s_0}$, $s_0 \approx 3.92$ 
(see Section~\ref{sec:Bounds3}). 
As $p_a^o K \ll M$,   then $p_a K > p_a^o K$ and $\cR_a$ decreases  with $p_a K$ and is maximal when $p_a K \sim M$, in which case $\cR_a = \frac{\tau_u - \tau_p}{\tau_u } \tau_p/ M$. This value is smaller than the maximal value found in the interval $p_a K \ll M$. 
For $p_a K \sim M$, the value of the sum rate is also smaller.


\vspace{2mm}
\subsubsection{$p_a K \ll \tau_p$}
The sum rate is approximated as
\begin{eqnarray}
\cR_a \approx
\frac{\tau_u - \tau_p}{\tau_u } p_a K \log_2\left( 1+ \frac{ \tau_p}{ K p_a }  \right).
\end{eqnarray}
For a fixed value of $\tau_p$, $\cR_a$ increases with $p_a K$ and takes its maximal value when $p_a K \sim \tau_p$, which case we consider next. 

\vspace{2mm}
\subsubsection{$p_a K = \alpha \tau_p$, $\alpha = O(1)$}
\begin{eqnarray}
\cR_a  \approx
\frac{\tau_u - \tau_p}{\tau_u }\tau_p \alpha \log_2 (1 + 1/\alpha).
\end{eqnarray}
The maximal value is smaller than the one obtained in the case $p_a K \gg \tau_p$.

\subsection{Case 2: $M \ll  \tau_u$}
\label{sec:Case2}

If $M \gg \tau_p$, then the results for Case 1 tells that the optimal value is of the form \eqref{eq:paK_opt} and \eqref{eq:tp_opt}.
This solution can be verified to give a smaller sum rate than the solution that is given when $M \ll \tau_p$ below.
The case $M \sim\tau_p$ is treated in Appendix~\ref{sec:Case3}.
We now assume $M \ll \tau_p$. Next we consider different definition intervals for the parameters.

\vspace{2mm}
\subsubsection{$p_a K \gg \tau_p$}
\begin{eqnarray}
\cR_a \approx
\frac{\tau_u - \tau_p}{\tau_u } p_a K \log_2\left( 1+ \frac{M \tau_p}{ p_a^2 K^2 }  \right).
\end{eqnarray}
As already mentioned in Appendix~\ref{sec:Case11}, for a fixed $\tau_p$, $\cR_a$ has its maximal value for $p_a^o K  =  \sqrt{M \tau_p}/\sqrt{3 s_0}$, $s_0 \approx 3.92$. 
As $p_a^o K \gg M$,   then $p_a K < p_a^o K$ and $\cR_a$ increases  with $p_a K$ and is maximal when $p_a K \sim \tau_p$, in which case $\cR_a = \frac{\tau_u - \tau_p}{\tau_u } M$. This value is smaller than the maximal value found in the interval $p_a K \ll \tau_p$. 

\vspace{2mm}
\subsubsection{$p_a K \ll \tau_p$}
\label{sec:Case22}

The SINR can be approximated as: 
\begin{eqnarray}
\underline{\SINR}_a \approx
 \frac{M \tau_p \beta_0^2}{ \overline{\beta}^2 p_a^2 K^2 +  \beta_0 \overline{\beta} \tau_p p_a K}.
\end{eqnarray}
With the same steps as in Case 1, $p_a K \gg \tau_p$, we obtain
\begin{eqnarray}
p_a^o K  = \sqrt{\frac{1}{2}   \frac{ \overline{\beta^4}}{\overline{\beta}^2 \overline{\beta^2}}} \sqrt{M \tau_p} + 
O(M ).
\label{eq:paK_opt2}
\end{eqnarray}

For the optimization w.r.t. $\tau_p$, a first order Taylor approximation of $\cR_a$ leads to
\begin{eqnarray}
\cR_a \log(2) & \approx &
\frac{\tau_u - \tau_p}{\tau_u }
\left(\frac{M \tau_p \beta_0^2}
{ \overline{\beta}^2 p_a K  +  \beta_0 \overline{\beta} \tau_p }  \right) \\
%
 & \approx & \frac{\tau_u - \tau_p}{\tau_u }
M \,
\E_{\beta_0}\!
\left( \frac{\beta_0}{\overline{\beta}} - 
\frac{\sqrt{M} }{ \sqrt{\tau_p} }  
\right) \\
 & \approx &
\frac{\tau_u - \tau_p}{\tau_u }
M 
\left( 1  - 
\frac{\sqrt{M} }{ \sqrt{\tau_p} }  
\right)
\label{eq:Case2_Ra}
\end{eqnarray}
Taking the derivative, we find that the optimal point  $\tau_p^o$  is the positive root of 3rd order polynomial equation: 
\begin{eqnarray}
- 2 x^{3/2}+ \alpha \sqrt{M} x  +  \tau_u \alpha\sqrt{M} = 0.
\end{eqnarray}
Writing $\tau_p$ as $\tau_p = a \tau_u$, we get:  
 $- 2 \frac{\tau_u}{M}  a^{3/2}+  a  +  1  = 0$. 
 
Assuming $a \ll 1$ and $a \ll \frac{\tau_u}{M}  a^{3/2}$, we obtain  $a = \left( \frac{M}{2 \tau_u} \right)^{2/3}$. 
One can further prove that values of $a$ that are not close to 0  lead to invalid solutions. Finally, the solution is:
\begin{eqnarray}
\tau_p^o = \left( \frac{M }{2 \tau_u} \right)^{2/3}    \tau_u  + o(\tau_u).
\end{eqnarray}

From (\ref{eq:Case2_Ra}), the sum rate can be written as:
\begin{eqnarray}
 \cR_a^o \log(2) = 
 M 
 + 
 O \left(
  M \left[   \frac{ M}{\tau_u} \right]^{2/3} 
 \right).
\end{eqnarray}

\subsubsection{$p_a K = \alpha  \tau_p$, $\alpha = O(1)$}
\begin{eqnarray}
\cR_a \! \! \!\! & \approx &  \! \!\!\! \frac{\tau_u - \tau_p}{\tau_u }  \alpha  \tau_p\log_2\left( 1+ \frac{M \tau_p}{   \alpha^2   \tau_p^2 +\alpha  \tau_p^2 }  \right) \\
&\approx & \! \! \!\!
\frac{\tau_u - \tau_p}{\tau_u } 
\alpha
\left[
\frac{M}{\alpha +1}  -\frac{1}{2}\frac{M^2}{(\alpha +1)^2 \tau_p}
\right] / \log(2).
\end{eqnarray}
One can prove that the associated maximal value is
smaller than the one obtained in the case $p_a K \ll \tau_p$.

\subsection{Case 3: ${M}$ and ${\tau_u}$ are of same order and  ${M}/{\tau_u} = \delta$, $\delta = O(1)$.}
\label{sec:Case3}

We show that the solution is of the following form:
\begin{eqnarray}
\tau_p^o = a \tau_u ;
\quad
p_a^o K_a  = b \sqrt{M \tau_u}  
\end{eqnarray}
where  $a$ and $b$ are scalars of order 1 and are solution of the maximization of the  function:
\begin{eqnarray}
\cR_{a}  \! =  \!
(1 - a ) b  \,
\E_{\beta_0}  \!\left[ \log_2   \!\!\left ( \!1+  
  \frac{ a {\beta}_0^2  \delta}{
       b \overline{\beta^2}  \delta \sqrt{\delta} +
               b^2 \overline{\beta}^2 \delta      +
        a b  \overline{\beta}  {\beta}_0  \sqrt{\delta}  }
\right) \right]  \!.
\label{eq:ab}
\end{eqnarray}
The expression \eqref{eq:ab} is obtained from the expression of $\cR_a$ in \eqref{eq:cRa}. One key point is that is only depends on ${M}/{\tau_u} = \delta$. 

We show that $a$ and $b$ are of order 1 provided that $\delta$ is of order 1 through numerical evaluations. 
In Figure~\ref{fig:figab}, the optimized values of $a$ and $b$ are shown as a function of $\delta$.   
For values of $\delta$ around 1, the values of $a$ and $b$ are indeed of order 1. 
When $\delta$ tends to $\infty$,  the value saturates to $a=b=1/2$ as predicted in Appendix~\ref{sec:Case1}, 
while the values obtained when  $\delta$ tends to $0$ correspond to Appendix~\ref{sec:Case22}.

\subsection{$\tau_p$ constrained by the channel coherence time: $\tau_p \leq \tau_p^{\max}$.}
\label{sec:Case4}

For cases  $M\gg \tau_p^{\max}$ and  $M \ll \tau_p^{\max}$, using equations  \eqref{eq:paK_opt} and  \eqref{eq:paK_opt2}
\begin{eqnarray}
p_a^o K  = \sqrt{\frac{1}{2}   \frac{ \overline{\beta^4}}{\overline{\beta}^2 \overline{\beta^2}}} \sqrt{M \tau_p^{\max}} + 
O(M ).
\end{eqnarray}
For case  $M \sim\tau_p^{\max}$, $p_a^o K_a  = b \sqrt{M \tau_p^{\max}}$  where $b$ optimizes 
\begin{eqnarray}
\cR_{a} = 
b  \;
\E_{\beta_0} \left[ \log_2 \left (1+  
  \frac{  {\beta}_0^2  \delta'}{
       b \overline{\beta^2}  \delta \sqrt{\delta'} +
               b^2 \overline{\beta}^2 \delta'      +
         b  \overline{\beta}  {\beta}_0  \sqrt{\delta'}  }
\right) \right].
\label{eq:ab2}
\end{eqnarray}
with $\delta' = M/\tau_p^{\max}$.

\begin{figure}[t]
\centering
\includegraphics[width=7cm]{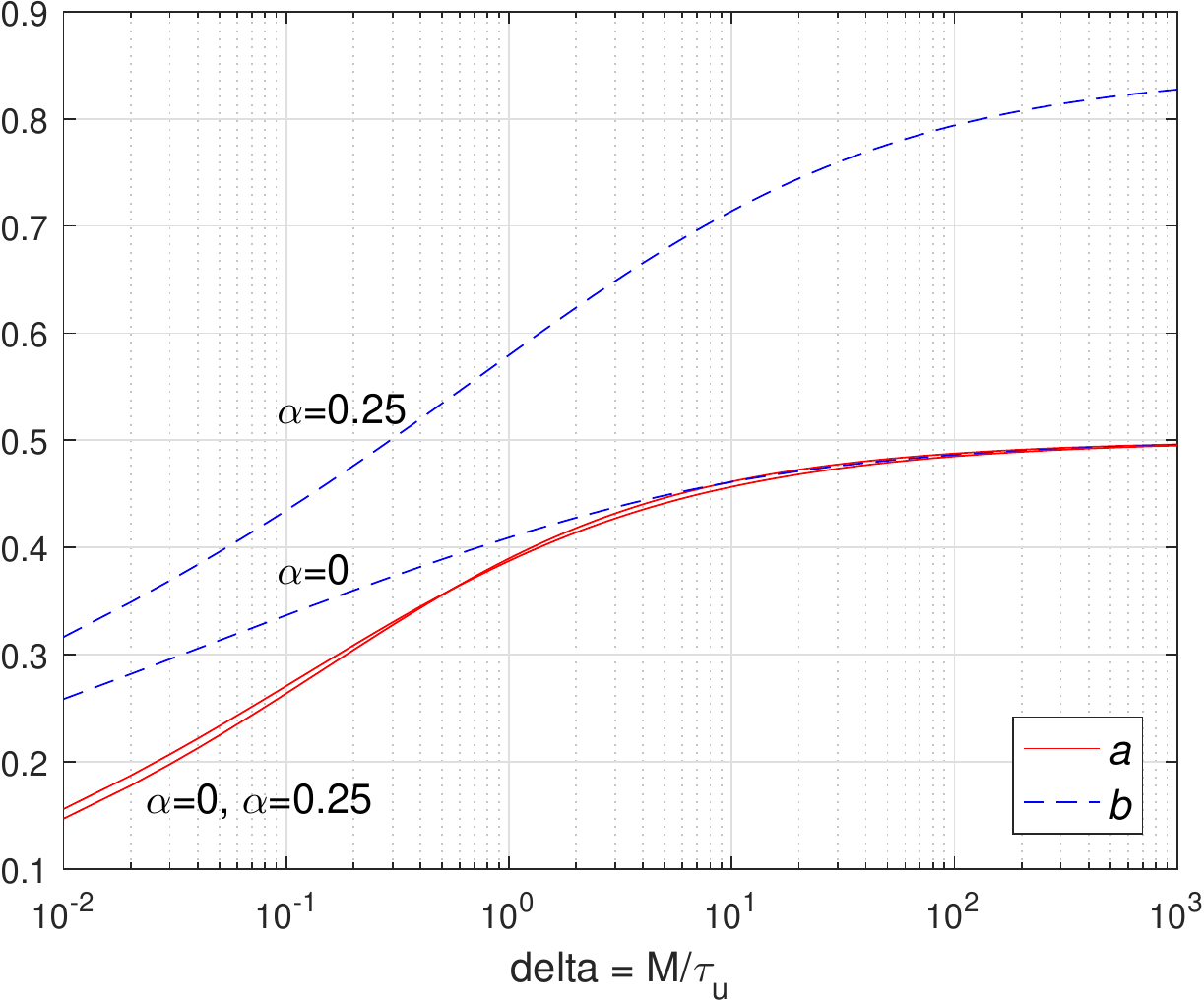}
\caption{Optimization of parameters a and b in Equation~(\ref{eq:ab2}) as a function of $\delta$. Model 3 with $\alpha=0$ and $\alpha=0.25$.}
\label{fig:figab}
\end{figure}

\bibliographystyle{IEEEbib}
\bibliography{strings,IEEEabrv,refs}

\end{document}